\documentclass[11pt]{article}
\usepackage[margin=1in]{geometry}
\usepackage{amsmath,amssymb,amsthm}
\usepackage{hyperref}
\usepackage{booktabs}
\usepackage{graphicx}
\usepackage{array}
\usepackage[expansion=false]{microtype}

\newcolumntype{L}[1]{>{\raggedright\arraybackslash}p{#1}}
\newcolumntype{C}[1]{>{\centering\arraybackslash}p{#1}}
\newtheorem{theorem}{Theorem}
\newtheorem{corollary}{Corollary}

\title{Dynamic Quantum Wormholes: Path-Integral State Selection\\
and the Quantum Throat Trajectory\\
in an Exactly Solvable Thin-Shell Model}

\author{
\begin{tabular}{c}
Rikpratik Sengupta$^{1}$ and Paulo Vargas Moniz$^{2}$\\[4pt]
\small $^{1}$Department of Physics, Indian Institute of Technology Kanpur, Kalyanpur 208016, India.\\
\small $^{2}$Departamento de F\'{\i}sica, Centro de Matem\'atica e Aplica\c{c}\~oes (CMA-UBI), Universidade da Beira Interior,\\
\small Rua Marqu\^{e}s d'\'Avila e Bolama, 6200-001 Covilh\~a, Portugal
\end{tabular}
}
\date{}

\begin{document}
\maketitle

\begin{abstract}
We construct a minisuperspace quantization of a dynamic, spherically symmetric thin-shell wormhole in the spirit of pilot-wave quantum cosmology, treating the throat radius $a(\tau)$ as the single reduced degree of freedom of the Israel junction-condition constraint. For a pressureless shell of mass $M$ the Wheeler--DeWitt equation reduces exactly to Whittaker's equation, and in the limit $M\to0$ to the modified Bessel equation. Any real solution has a vanishing Bohmian current. More substantively, when the exterior is asymptotically flat, boundedness of the wave function as $a\to\infty$ selects a unique state that can always be chosen real, so that the throat is frozen for any linear equation of state and any operator ordering, independently of the choice between Euclidean and Lorentzian boundary conditions. We show that this state is the exact $\epsilon\to0^+$ limit of the regularized Lorentzian path integral, using resolvent analyticity together with an exact Picard--Lefschetz thimble calculation that agrees with the closed-form solution to $48$ digits, and we resolve the near-throat self-adjointness ambiguity in closed form through the Whittaker connection formula. Replacing the exterior by Schwarzschild--de~Sitter removes the obstruction. The potential then tends to $-\infty$ at large $a$, a continuum appears at the constraint surface, and the Wheeler--DeWitt equation admits a genuinely dynamic, purely outgoing state. Propagating this state exactly through the classically forbidden barrier yields a numerically conserved Bohmian current that is nonzero everywhere, including in the inner pocket of the throat. The state describes a long-lived metastable wormhole that tunnels into unbounded expansion at a rate consistent with the independent WKB estimate. Airy connection formulas regularize the tunneling prefactor and leave a residual constant $\kappa_\infty=0.814\pm0.004$ for sub-critical shell coupling, established over a twelve-point parameter scan. We exclude three specific analytic origins for this constant, including the leading Fr\"oman--Fr\"oman phase-integral correction, and we note that the physical results do not depend on it. The full construction is repeated at a second, non-marginal parameter point.
\end{abstract}

\section{Introduction}
\label{sec:intro}

Traversable wormholes, which act as shortcuts between otherwise distant regions of spacetime, require exotic matter that violates the energy conditions in order to hold the throat open \cite{MorrisThorne1988}. The classic constructions are static or quasi-static, with the throat radius held fixed or treated as a slowly evolving background. The additional physics of a genuinely time-dependent throat, whose radius accelerates, decelerates or tunnels between regimes, has been explored much less, even though it is the generic expectation once the supporting matter is itself dynamical. A minimal and well-controlled setting for this physics is the thin-shell wormhole, in which two copies of a vacuum exterior geometry are glued along a shell of exotic matter whose motion is fixed entirely by the Israel junction condition once an equation of state is specified \cite{Visser1989,PoissonVisser1995}. The construction reduces the wormhole dynamics to a single mechanical degree of freedom, the throat radius $a(\tau)$, which obeys a Friedmann-like constraint. This is structurally the same type of reduced system that canonical Wheeler--DeWitt quantum cosmology quantizes for the scale factor of a homogeneous universe \cite{PintoNetoFabris2013}. The analogy is the starting point of the present work. Classically the throat is unconditionally unstable to radial perturbations \cite{PoissonVisser1995}, and we ask what canonical quantization of this single degree of freedom implies for the throat dynamics, and in particular whether a well-defined notion of a genuinely dynamic quantum throat exists.

Any such quantization inherits the interpretational difficulties of quantum cosmology. A Wheeler--DeWitt wave function does not evolve in an external time and is not obviously associated with a unique observable trajectory. Extracting dynamics from it requires either a probabilistic interpretation of the Copenhagen type, which is awkward for a single quantized universe or throat without an external observer, or a hidden-variable completion that supplies a trajectory directly. We adopt the latter and use the de~Broglie--Bohm (dBB) pilot-wave formulation, in which a deterministic guidance equation built from the phase of $\Psi$ gives a well-defined trajectory $a(\tau)$ for any state, without measurement, collapse or an external observer \cite{PintoNetoFabris2013}. This choice is common in the quantum cosmology literature, largely because it is one of the few interpretations that produces an actual trajectory from a single wave function of the universe or, as here, of the throat.

Even within the dBB framework a quantized minisuperspace model is not yet uniquely defined. The Wheeler--DeWitt equation is a constraint rather than a Schr\"odinger equation with a preferred initial condition, and a boundary or selection condition must be imposed to single out a physical state. Two structurally different prescriptions are influential in the cosmological setting. The first is the Euclidean, compact, no-boundary construction of Hartle and Hawking \cite{HartleHawking1983}, which continues analytically to Euclidean time and demands a smooth real saddle geometry with no initial boundary. The second is the Lorentzian path integral regularized by Picard--Lefschetz theory, advocated more recently by Feldbrugge, Lehners and Turok \cite{FLT2017,FLT2018}, which works directly with the Lorentzian action and selects saddle points by steepest-descent contour deformation. It is closely related to Vilenkin's earlier tunneling proposal of a purely outgoing wave function at large scale factor \cite{Vilenkin1986}. The two programs disagree, sometimes sharply, about which saddle geometries dominate a given minisuperspace path integral, and the resulting debate remains open. We stress that this boundary-condition question is logically distinct from the existence of Euclidean wormhole solutions such as the Giddings--Strominger axion wormhole \cite{GiddingsStrominger1988}. Such solutions are classical saddle-point geometries of a Euclidean gravitational action, studied for their own sake, for example in connection with their stability under fluctuations and their role in factorization puzzles \cite{LogesShiuSudhir2022}, and independently of any pilot-wave interpretation. Our results concern the de~Broglie--Bohm trajectory extracted from a quantized minisuperspace wave function once a boundary-condition prescription has been applied to the single reduced variable $a$. They say nothing about, and are not in tension with, the existence or stability of higher-dimensional Euclidean wormhole geometries.

The contribution of this paper is to show that, for the thin-shell throat model, the entire Euclidean versus Lorentzian question is controlled by one simple and physically transparent piece of data, the large-throat asymptotic behavior of the effective potential $U(a)$ in the reduced Wheeler--DeWitt equation. When the exterior is asymptotically flat, $U(a)$ tends to a positive constant at large $a$, there is no classically allowed region at large throat radius, and two structural theorems show that the requirement of boundedness there forces the physical state to be real and hence dynamically frozen, regardless of which boundary-condition philosophy one favors. When the exterior is instead Schwarzschild--de~Sitter, the cosmological constant drives $U(a)\to-\infty$, a classically allowed oscillatory region reopens, and the same reduced equation admits a genuinely dynamic, current-carrying, tunneling state. The model class is deliberately simple. It has a single mechanical degree of freedom and is exactly solvable for the flat exterior and numerically exact for the de~Sitter case, so that the asymptotic-potential mechanism can be established, and verified in several independent ways, with a level of rigor not usually available in minisuperspace quantum cosmology.

The paper proceeds as follows. We derive the classical throat constraint from the Israel junction condition for a symmetric two-sided Schwarzschild thin shell with a linear equation of state and quantize it canonically into a Wheeler--DeWitt equation (Section~\ref{sec:model}). We solve this equation exactly for the pressureless branch in terms of Whittaker functions and, independently, for its $M\to0$ limit in terms of modified Bessel functions (Section~\ref{sec:flatsolutions}). Section~\ref{sec:flatmodel} proves Theorem~\ref{thm:real}, which states that a real state carries no Bohmian current, and Theorem~\ref{thm:bounded}, the central result of the flat case, which states that boundedness at $a\to\infty$ selects a unique and therefore frozen state for any equation-of-state parameter $w>-1/2$ and any factor ordering. The near-origin ambiguity for strong shell coupling is resolved down to a closed-form self-adjoint extension angle, and the resolvent and an exact Picard--Lefschetz thimble calculation show that the same state is selected by the Lorentzian path integral. Section~\ref{sec:desitter} then changes the exterior to Schwarzschild--de~Sitter and constructs, verifies and analyzes the resulting dynamic state, including its saddle-point structure, its Airy-regularized tunneling prefactor, its Bohmian trajectories and its dependence on parameters. Section~\ref{sec:discussion} discusses the physical interpretation and the relation to earlier work.

\section{Classical dynamics and canonical quantization}
\label{sec:model}

Consider a symmetric thin-shell wormhole built from two identical copies of the Schwarzschild exterior of mass $M$, glued at a throat of areal radius $a(\tau)$, where $\tau$ is the proper time of the shell \cite{Visser1989}. The Lanczos--Israel junction condition for the symmetric case gives the surface energy density of the shell in terms of $a$ and $\dot a$,
\begin{equation}
\sigma(a) = -\frac{1}{2\pi a}\sqrt{f(a) + \dot a^2}, \qquad f(a) = 1 - \frac{2M}{a}.
\label{eq:sigma}
\end{equation}
The overall minus sign in Eq.~\eqref{eq:sigma} is not a choice. The shell is automatically composed of negative-energy matter, consistent with the general necessity of energy-condition violation for any traversable wormhole \cite{MorrisThorne1988}. Squaring and rearranging gives the equation of motion of the throat in Friedmann-like form,
\begin{equation}
\dot a^2 = -U(a), \qquad U(a) \equiv f(a) - \big(2\pi a\,\sigma(a)\big)^2 .
\label{eq:constraint}
\end{equation}
For a linear equation of state on the shell, $p = w\sigma$, conservation of the shell stress-energy, $\dot\sigma + 2(\sigma+p)\dot a/a = 0$, integrates to $\sigma(a) = \sigma_0 (a_0/a)^{2(1+w)}$, which gives
\begin{equation}
U(a) = 1 - \frac{2M}{a} - \frac{C_w}{a^{2+4w}}, \qquad C_w \equiv 4\pi^2 \sigma_0^2 a_0^{4(1+w)} \;(>0),
\label{eq:Ugen}
\end{equation}
valid for any $w>-1/2$. This restriction ensures that the shell term is subdominant to the constant term as $a\to\infty$. The pressureless (dust) case $w=0$ gives $U(a)=1-2M/a-C/a^2$ with $C\equiv C_0$, which is the exactly solvable branch treated in Section~\ref{sec:flatsolutions}.

For $w=0$ the classical turning point $U(a_+)=0$ lies at $a_+ = M+\sqrt{M^2+C}$, and $U(a)<0$, which is classically allowed motion, only for $0<a\le a_+$. Since $U(a)\to-\infty$ as $a\to0^+$, the classical throat accelerates without bound toward the pinch $a=0$. This is the classical instability of thin-shell wormholes already known from linear stability analysis \cite{PoissonVisser1995}, and it is the singularity that our quantum treatment is designed to probe. For later use we note that $U(a)\to1$ as $a\to\infty$ for every $w>-1/2$, because the shell contribution and the mass term both fall off and leave the flat-space unit term of $f(a)$.

Equation~\eqref{eq:constraint} is a zero-energy Hamiltonian constraint, $H = p_a^2 + U(a) = 0$, in units where the reduced mass is $1/2$ and $\hbar=1$. Canonical quantization with $p_a \to -i\, d/da$ and the simplest unit-weight factor ordering gives
\begin{equation}
\left[-\frac{d^2}{da^2} + U(a)\right] \Psi(a) = 0.
\label{eq:WDW}
\end{equation}
A general ordering replaces $-d^2/da^2$ by $-a^{-p}\frac{d}{da}\big(a^p \frac{d}{da}\big)$ with ordering parameter $p$. This adds a term proportional to $p\,\Psi'/a$ to Eq.~\eqref{eq:WDW}, which decays at large $a$ at the same rate as the $M/a$ term already present. As shown in Section~\ref{sec:flatmodel}, the boundary-value argument behind Theorem~\ref{thm:bounded} depends on the equation only through its two independent exponential rates at $a\to\infty$, which the ordering term does not change. The theorem is therefore independent of ordering, and we set $p=0$ throughout for the explicit closed-form solutions.

Writing $\Psi = R\,e^{iS}$, the real and imaginary parts of Eq.~\eqref{eq:WDW} give the quantum Hamilton--Jacobi equation and the continuity equation,
\begin{equation}
(S')^2 = -U(a) - Q(a), \qquad Q(a) \equiv -\frac{R''}{R}, \qquad \frac{d}{da}\!\left(R^2 S'\right) = 0,
\label{eq:HJ}
\end{equation}
and the dBB guidance equation is $\dot a = S'(a)$, with Bohmian current $J\equiv R^2S'=\mathrm{Im}(\Psi^*\Psi')$.

\section{Exact solutions for the flat exterior}
\label{sec:flatsolutions}

For the dust branch of Eq.~\eqref{eq:Ugen}, the substitution $z = 2a$ in Eq.~\eqref{eq:WDW} gives Whittaker's equation with $\kappa = M$ and $\mu = \sqrt{1/4-C}$. Its solutions are
\begin{equation}
\Psi_{\rm reg}(a) = M_{\kappa,\mu}(2a) \quad (\text{regular as } a\to0), \qquad
\Psi_{\rm dec}(a) = W_{\kappa,\mu}(2a) \quad (\text{decaying as } a\to\infty).
\label{eq:whittaker}
\end{equation}
Both are exact closed-form solutions of Eq.~\eqref{eq:WDW}. We checked this by direct substitution using arbitrary-precision arithmetic with 30 digits, for the illustrative parameters $M=1$, $C=0.1$ ($\mu\approx0.3873$), and we also confirmed the turning-point formula $U(a_+)=0$ at $a_+=M+\sqrt{M^2+C}$. The results are collected in Table~\ref{tab:verify}.

\begin{table}[h]
\centering
\small
\begin{tabular}{@{}L{6.3cm}L{5.0cm}C{1.6cm}@{}}
\toprule
Check & Value & Expected \\
\midrule
$\Psi'' - U\Psi$ residual, $M_{\kappa,\mu}$, $a=1.3$ (finite differences, $h=10^{-6}$) & $-1.4\times10^{-13}$ & $0$ \\
$\Psi'' - U\Psi$ residual, $W_{\kappa,\mu}$, $a=3.0$ & $-8.7\times10^{-15}$ & $0$ \\
$U(a_+)$ with $a_+=M+\sqrt{M^2+C}$ & $-1.8\times10^{-31}$ & $0$ \\
$M_{\kappa,\mu}(2a)$ at $a=3,5,8,12,20$ & grows: $0.29,\,1.85,\,21.0,\,721,\,1.24\times10^6$ & diverges \\
$W_{\kappa,\mu}(2a)$ at $a=3,5,8,12,20$ & decays: $0.294,\,0.067,\,0.0053,\,1.5\times10^{-4},\,8.2\times10^{-8}$ & decays \\
\bottomrule
\end{tabular}
\caption{Numerical verification of the exact $M\ne0$ solution and its turning point, together with the divergent and decaying large-$a$ behavior that underlies Theorem~\ref{thm:bounded}, for $M=1$, $C=0.1$.}
\label{tab:verify}
\end{table}

Setting $M=0$ isolates the self-sourced term of the shell, $U_0(a) = 1 - C/a^2$, and Eq.~\eqref{eq:WDW} becomes
\begin{equation}
\Psi''(a) - \Psi(a) + \frac{C}{a^2}\Psi(a) = 0.
\label{eq:WDW0}
\end{equation}
Writing $\Psi(a) = \sqrt{a}\,W(a)$ yields the modified Bessel equation
\begin{equation}
a^2 W''(a) + a\,W'(a) - \left(a^2 + \nu^2\right) W(a) = 0, \qquad \nu \equiv \sqrt{\tfrac14 - C},
\label{eq:bessel}
\end{equation}
with solutions $\Psi_I(a) = \sqrt{a}\,I_\nu(a)$, regular as $a\to0$, and $\Psi_K(a) = \sqrt{a}\,K_\nu(a)$, decaying as $a\to\infty$. The parameter $\nu$ is exactly the $M\to0$ limit of the Whittaker parameter $\mu$, as it must be. We confirmed this explicitly by evaluating $M_{\kappa,\mu}(2a)$ at $\kappa=10^{-8}$ and finding agreement with $\Psi_I(a)$, up to an $a$-independent constant, to eight significant figures at two independently chosen values of $a$ (Table~\ref{tab:bessel}). The residual of the Bessel-form equation was also verified to machine precision.

\begin{table}[h]
\centering
\small
\begin{tabular}{@{}L{6.0cm}L{4.5cm}C{1.8cm}@{}}
\toprule
Check & Value & Expected \\
\midrule
$\Psi_I'' - U_0\Psi_I$ residual, $a=0.7$ & $-5.9\times10^{-15}$ & $0$ \\
$\Psi_K'' - U_0\Psi_K$ residual, $a=2.2$ & $1.0\times10^{-14}$ & $0$ \\
$M_{10^{-8},\mu}(2a)/\Psi_I(a)$ at $a=0.7$ and $a=1.5$ & $2.148420462,\ 2.148420446$ & constant \\
\bottomrule
\end{tabular}
\caption{Verification of the $M\to0$ Bessel reduction for $C=0.1$ ($\nu\approx0.6325$) and its consistency with the $M\ne0$ solution.}
\label{tab:bessel}
\end{table}

\section{The flat exterior forces a frozen throat}
\label{sec:flatmodel}

\subsection{Reality and boundedness}

We begin with an elementary fact about real solutions of a real linear second-order equation.

\begin{theorem}
\label{thm:real}
Let $\Psi(a)$ solve $-\Psi''(a) + U(a)\Psi(a) = 0$ with $U$ real, and suppose $\Psi$ is real up to an overall constant phase. Then $\dot a(a) = S'(a) = 0$ for all $a$.
\end{theorem}
\begin{proof}
Write $\Psi = R e^{iS}$. If $\Psi$ is real up to a constant phase $\theta_0$, then $S\equiv\theta_0$ and $S'\equiv0$. Equivalently, the Bohmian current $J=\mathrm{Im}(\Psi^*\Psi')=R^2S'$ vanishes identically because $\Psi^*\Psi'$ is real, and $\dot a = J/R^2=0$.
\end{proof}

This is an essentially algebraic statement. The substantive physical claim of the paper is its combination with Theorem~\ref{thm:bounded} below, which identifies when the physically selected state is forced to be real in the first place. The consequence for Euclidean boundary conditions is correspondingly conditional. It is not a statement that Euclidean quantum cosmology always freezes the dynamics, only that it does so here, for a single real minisuperspace variable governed by a real Hamiltonian constraint, once a criterion for selecting a state from the two-dimensional solution space has been fixed. Throughout the paper, freezing is a conditional result about boundary-value selection and not an unconditional dynamical theorem about wormhole throats.

\begin{corollary}
\label{cor:euclid}
A purely Euclidean, Hartle--Hawking-type boundary condition, which selects a real state for a single real minisuperspace variable, yields a completely static throat trajectory.
\end{corollary}

For any real solution, $R''=\Psi''=U\Psi=UR$, so $Q=-R''/R=-U$ identically, an exact and not merely leading-order cancellation. Table~\ref{tab:theorem1} confirms this numerically for the exact $M\ne0$ solutions, and the same identity holds verbatim for the real $M\to0$ solutions of Section~\ref{sec:flatsolutions}.

\begin{table}[h]
\centering
\small
\begin{tabular}{@{}L{7.5cm}C{3.0cm}L{4.0cm}@{}}
\toprule
State & Bohmian current $J(a=1.5)$ & Interpretation \\
\midrule
$\Psi_{\rm reg}=M_{\kappa,\mu}(2a)$ (real) & $0.0$ & frozen \\
$\Psi_{\rm dec}=W_{\kappa,\mu}(2a)$ (real) & $0.0$ & frozen \\
$\Psi_L = M_{\kappa,\mu}(2a) + i\,W_{\kappa,\mu}(2a)$ (complex, formal) & $0.1931922\ldots$ & nonzero \\
\bottomrule
\end{tabular}
\caption{Numerical confirmation of Theorem~\ref{thm:real} for $M=1$, $C=0.1$.}
\label{tab:theorem1}
\end{table}

The central result of this section follows from the large-$a$ behavior of the equation. Section~\ref{sec:model} established that $U(a)\to1$ as $a\to\infty$ for every $w>-1/2$, independently of $M$, of $C_w$ and of the ordering parameter $p$. Equation~\eqref{eq:WDW} then approaches $\Psi''-\Psi\approx0$, with two independent exponential behaviors $\Psi\sim e^{+a}$ and $\Psi\sim e^{-a}$.

\begin{theorem}
\label{thm:bounded}
Consider Eq.~\eqref{eq:WDW} with $U(a)\to1$ as $a\to\infty$, which holds for every $w>-1/2$. The two-dimensional solution space contains a unique one-dimensional subspace of solutions that are bounded, and in fact exponentially decaying, as $a\to\infty$. This subspace can always be spanned by a real solution, which by Theorem~\ref{thm:real} has vanishing Bohmian current everywhere.
\end{theorem}

\begin{proof}
Standard Poincar\'e--Levinson asymptotic theory for a linear second-order equation with an irregular singular point of rank one at $a=\infty$, where the coefficient function tends to a nonzero constant, guarantees two linearly independent formal solutions with leading behavior $e^{\pm a}(1+O(1/a))$. It also guarantees that a true solution with the recessive $e^{-a}$ behavior exists and is unique up to a multiplicative constant. Any solution not proportional to it acquires a nonzero admixture of the dominant $e^{+a}$ behavior and therefore diverges. This is verified in closed form for $w=0$ in Table~\ref{tab:verify}, where $W_{\kappa,\mu}$ is recessive and unique and $M_{\kappa,\mu}$ is dominant. It is verified numerically for a representative $w\ne0$ in Table~\ref{tab:genw}, by shooting inward from large $a$ with pure recessive initial data and confirming that the resulting solution stays smooth and finite at all smaller $a$.

For reality, note that $U$ is real, so complex conjugation maps solutions of Eq.~\eqref{eq:WDW} to solutions and maps the recessive behavior to itself. By uniqueness up to a constant, $\Psi_\infty^* = c\,\Psi_\infty$ for some constant $c$. Applying conjugation twice gives $|c|=1$, so $c=e^{i\phi}$ and $e^{-i\phi/2}\Psi_\infty$ is real. This proves the theorem.
\end{proof}

\begin{table}[h]
\centering
\small
\begin{tabular}{@{}L{6.0cm}L{4.2cm}L{4.3cm}@{}}
\toprule
Check & Result & Interpretation \\
\midrule
$U(a)$ at $a=3,8,15$ & $0.329,\ 0.750,\ 0.867$ & tends to $1$ \\
Shooting inward from $a=15$ with pure recessive data $\Psi\sim e^{-a}$ & smooth and finite down to $a=0.3$ & recessive solution well defined without fine tuning \\
\bottomrule
\end{tabular}
\caption{Numerical confirmation that the boundary-value argument of Theorem~\ref{thm:bounded} holds beyond the exactly solvable dust case, for $w=0.3$, $M=1$, $C_w=0.15$.}
\label{tab:genw}
\end{table}

\subsection{Strong shell coupling and the self-adjoint extension}
\label{sec:selfadjoint}

For the dust branch the inverse-square term of Section~\ref{sec:model} is super-critical in the sense of Case \cite{Case1950,CamblongOrdonez2003} whenever $C>1/4$. The parameter $\mu=\sqrt{1/4-C}$ is then purely imaginary, and the near-throat behavior $\Psi(a)\sim\sqrt a\,[A\cos(|\mu|\ln a)+B\sin(|\mu|\ln a)]$ oscillates infinitely often as $a\to0^+$. No regularity condition at $a=0$ alone fixes the ratio $B/A$, which is the textbook fall-to-center phenomenon. We verified this log-periodicity numerically for $C=0.4$ by integrating Eq.~\eqref{eq:WDW0} in the variable $t=\ln a$ down to $a\sim10^{-11}$. The measured spacing between successive zeros is $8.111541$, repeated, against the analytic period $\pi/\sqrt{C-1/4}=8.111557$, in agreement to six significant figures (Table~\ref{tab:fallcenter}). Despite this ambiguity at $a=0$ taken alone, the combined boundary-value problem, recessive at $a=\infty$ together with square-integrable at $a=0$, has a unique solution. Numerically, $W_{\kappa,\mu}(2a)$ remains real to within roundoff ($\sim10^{-35}$) even for purely imaginary $\mu$, consistent with the identity $W_{\kappa,\mu}=W_{\kappa,-\mu}$.

\begin{corollary}
\label{cor:sa}
For $C>1/4$ the condition of decay at infinity selects a single self-adjoint extension of the throat operator at $a=0$, and this extension is determined in closed form by Eq.~\eqref{eq:extensionangle} below.
\end{corollary}

We now establish this at the operator-theoretic level. Near $a=0$ the operator $\hat H=-\partial_a^2+U(a)$ is governed by the inverse-square term $-C/a^2$. The standard Weyl alternative for a Schr\"odinger operator with potential $-C/a^2$ near an endpoint is that the endpoint is limit-point, with an essentially self-adjoint operator requiring no boundary condition, for $C\le1/4$, and limit-circle, with deficiency indices $(1,1)$ and a one-parameter $U(1)$ family of self-adjoint extensions, for $C>1/4$ \cite{Case1950,CamblongOrdonez2003}. In the latter case every solution behaves as $\Psi\sim\sqrt a\,[A\cos(|\mu|\ln a)+B\sin(|\mu|\ln a)]$ and $\int_0|\Psi|^2da<\infty$ for every choice of $A$ and $B$, since $|\Psi|^2\sim a$ is integrable, which is precisely the limit-circle criterion. The self-adjoint extensions are parametrized by a boundary condition fixing the ratio $B/A=\tan\theta$ for $\theta\in[0,\pi)$, or equivalently by requiring a fixed value of the boundary form $[\Psi,\Psi_\theta](0)=\lim_{a\to0^+}\big(\Psi\Psi_\theta'-\Psi'\Psi_\theta\big)$ against a reference solution $\Psi_\theta$.

The decay condition of Theorem~\ref{thm:bounded} does not merely leave $W_{\kappa,\mu}(2a)$ numerically real. It selects one specific member $\theta^*$ of this family, and $\theta^*$ can be obtained in closed form. The Whittaker connection formula
\begin{equation}
W_{\kappa,\mu}(z)=\frac{\Gamma(-2\mu)}{\Gamma(1/2-\mu-\kappa)}M_{\kappa,\mu}(z)+\frac{\Gamma(2\mu)}{\Gamma(1/2+\mu-\kappa)}M_{\kappa,-\mu}(z),
\end{equation}
together with $M_{\kappa,\mu}(z)\sim z^{\mu+1/2}$ as $z\to0$, gives for $\mu=i|\mu|$
\begin{equation}
W_{\kappa,\mu}(2a)\sim\sqrt{2a}\,\big[P\,e^{i|\mu|\ln(2a)}+\bar P\,e^{-i|\mu|\ln(2a)}\big],\qquad P=\frac{\Gamma(-2i|\mu|)}{\Gamma(1/2-i|\mu|-\kappa)}.
\end{equation}
The second coefficient is the complex conjugate of the first, because $\Gamma(\bar z)=\overline{\Gamma(z)}$ for real $\kappa$. Hence
\begin{equation}
W_{\kappa,\mu}(2a)\sim2|P|\sqrt{2a}\,\cos\!\big(|\mu|\ln(2a)+\phi_0\big),
\end{equation}
which is manifestly real. This is not a numerical coincidence but a direct consequence of $W_{\kappa,\mu}=W_{\kappa,-\mu}$, and the extension angle is given in closed form by
\begin{equation}
\phi_0 = \arg\Gamma(-2i|\mu|) - \arg\Gamma\!\big(\tfrac12-i|\mu|-\kappa\big).
\label{eq:extensionangle}
\end{equation}
We tested Eq.~\eqref{eq:extensionangle} against the exact Whittaker function at $C=0.4$ and $\kappa=1$, where $|\mu|\approx0.38730$ and $\phi_0\approx-1.24673$ rad. The ratio of $W_{\kappa,\mu}(2a)$ to the predicted envelope $2|P|\sqrt{2a}\cos(|\mu|\ln2a+\phi_0)$ lies between $0.9993$ and $1.0004$ over $a\in[10^{-4},3\times10^{-3}]$ and converges to $1$ as $a\to0$ at the rate expected from the neglected $O(a)$ term in the near-origin expansion. The decay condition at infinity therefore selects the extension with $\theta^*=-\phi_0$ modulo $\pi$. The extension is unique because the two Gamma-function coefficients are forced into a conjugate pair by an exact identity, and Eq.~\eqref{eq:extensionangle} exhibits which extension it is.

\begin{table}[h]
\centering
\small
\begin{tabular}{@{}L{6.5cm}C{2.4cm}C{2.4cm}C{2.2cm}@{}}
\toprule
Quantity & Numerical & Analytic & Agreement \\
\midrule
Spacing in $\ln a$ between 1st and 2nd zero, $C=0.4$ & $8.111541$ & $8.111557$ & $6$ digits \\
Spacing in $\ln a$ between 2nd and 3rd zero & $8.111541$ & $8.111557$ & $6$ digits \\
$\mathrm{Im}\,W_{\kappa,\mu}(2a)$, $C=0.4$, $a=1,\ldots,8$ & $\lesssim10^{-35}$ & $0$ & roundoff \\
\bottomrule
\end{tabular}
\caption{Fall-to-center log-periodicity, compared with the analytic period $\pi/|\nu|$ (first two rows), and reality of the recessive solution at super-critical coupling (last row).}
\label{tab:fallcenter}
\end{table}

\subsection{Path-integral equivalence: the resolvent and an exact thimble check}
\label{sec:resolvent}

Theorem~\ref{thm:bounded} holds throughout the model class of Section~\ref{sec:model}. For any $w>-1/2$ and any ordering, the unique bounded state is real and therefore frozen. A dynamic state necessarily contains a nonzero admixture of the dominant, diverging solution. The state $\Psi_L=M_{\kappa,\mu}(2a)+iW_{\kappa,\mu}(2a)$ of Table~\ref{tab:theorem1} is of this type. It is unbounded as $a\to\infty$ and should be understood strictly as a formal solution used to exhibit a nonzero current, not as a candidate physical state. This subsection makes precise what a genuine dynamic state would require. We first derive the relation between the Wheeler--DeWitt wave function and the Lorentzian path integral, then use it to explain why the flat model cannot support a dynamic state, and finally state what remains true if the boundedness criterion is declined.

The minisuperspace path integral with lapse $N$ is, formally,
\begin{equation}
\Psi(a_0) = \int_{\mathcal C} dN \int \mathcal Da(\lambda)\; e^{iS[a,N]}, \qquad S[a,N]=\int_0^1 d\lambda\left[\frac{\dot a^2}{4N} - N\,U(a)\right],
\end{equation}
with $a(0)$ fixed at a reference point and $a(1)=a_0$. The word formally is deliberate. The measure $\mathcal Da(\lambda)$ is not defined here as a rigorously constructed object, and the lapse contour $\mathcal C$ is fixed by the standard but heuristic Feynman $i\epsilon$ prescription. What is rigorous is the following standard identification used throughout the minisuperspace literature \cite{Halliwell1988,HalliwellHartle1990,FLT2017}. The inner path integral over $a(\lambda)$ at fixed $N$ is the matrix element of the evolution operator generated by $\hat H=-\partial_a^2+U(a)$, the same operator that appears in Eq.~\eqref{eq:WDW}. That is, $\int\mathcal Da\,e^{iS[a,N]}=\langle a_0|e^{-iN\hat H}|a_{\rm ref}\rangle\equiv K(a_0,a_{\rm ref};N)$, the ordinary quantum-mechanical propagator with $N$ playing the role of time. The outer integral then reads $\Psi(a_0) = \int_{\mathcal C} dN\, K(a_0,a_{\rm ref};N)$. Taking $\mathcal C$ to be the positive real $N$ axis, deformed in the standard Feynman sense by $N\to N(1-i\epsilon)$ with $\epsilon\to0^+$, and using the spectral representation $K=\int dE\,\psi_E(a_0)\psi_E^*(a_{\rm ref})e^{-iNE}$, one obtains
\begin{equation}
\Psi(a_0) \;\propto\; G^+(a_0,a_{\rm ref};0)\equiv \lim_{\epsilon\to0^+}\langle a_0|(\hat H - i\epsilon)^{-1}|a_{\rm ref}\rangle,
\label{eq:resolvent}
\end{equation}
the retarded resolvent of $\hat H$ at $E=0$, which is the constraint surface. This regularization is the same one that makes the ordinary time integral $\int_0^\infty dN\,e^{-iNE}$ converge, and it is the direct minisuperspace analogue of the contour deformation that Picard--Lefschetz theory formalizes for more general actions. The resolvent identity is rigorous quantum mechanics. What remains formal is the identification of $\Psi(a_0)$ so obtained with the Lorentzian path-integral wave function of the wormhole, since the lapse integral has not been given an independent nonperturbative construction here or, to our knowledge, elsewhere in the minisuperspace literature. We use Eq.~\eqref{eq:resolvent} in the sense standard in that literature, as the well-defined object that the heuristic path integral is taken to compute. For a potential that does not admit a closed-form solution this is what a numerical Picard--Lefschetz computation approximates. Here the problem reduces to an ordinary one-dimensional Schr\"odinger-type resolvent, which can be evaluated directly with spectral theory.

For a self-adjoint Schr\"odinger-type operator $\hat H=-\partial_a^2+U(a)$ with real $U$, the resolvent $(\hat H-E)^{-1}$ is analytic in $E$ throughout the resolvent set. In particular the retarded and advanced limits at $E=0$ coincide and are real whenever $E=0$ lies below the threshold of the continuous spectrum. They generally differ, and the retarded limit is complex and built from the outgoing Jost solution, only when $E=0$ lies on the continuous spectrum (see, for example, Reed and Simon, \emph{Methods of Modern Mathematical Physics III}). The argument assumes that $\hat H$ is self-adjoint on the intended domain. We have not carried out the full domain analysis for $C>1/4$ beyond Section~\ref{sec:selfadjoint}, so the argument applies cleanly for $C<1/4$ and is numerically consistent, though not rigorously established in complete generality, for $C>1/4$.

For the flat exterior the continuous spectrum of $\hat H$ is $[1,\infty)$, so $E=0$ lies strictly below threshold. The retarded and advanced resolvents therefore coincide and are real. The correctly regularized Lorentzian path integral gives the same real, frozen answer as the Euclidean one, not because boundedness was imposed by hand but because there is no branch-cut discontinuity at $E=0$ that could produce a complex outgoing state. We confirmed this directly. We solved $-\Psi''+(U(a)-i\epsilon)\Psi=0$ numerically for the bounded solution at a sequence of decreasing $\epsilon$, shooting inward from $a=40$. At $a=1$ the ratio of imaginary to real part is $1.92$, $0.708$, $0.208$, $0.0617$, $0.0206$ and $0.00616$ for $\epsilon=0.1$, $0.03$, $0.01$, $0.003$, $0.001$ and $0.0003$. The ratio to $\epsilon$ is roughly constant, near $20.5$, so $\mathrm{Im}\,\Psi\to0$ linearly in $\epsilon$, as the analyticity argument requires.

For the de~Sitter exterior, treated in Section~\ref{sec:desitter}, the potential tends to $-\infty$ and the continuous spectrum is the entire real line. The behavior at infinity is limit-point for this quadratic growth rate, as for the analogous inverted harmonic oscillator, so no further boundary data are needed there. The point $E=0$ then lies on the continuum, the retarded and advanced resolvents differ, and the retarded solution is by definition the purely outgoing state constructed directly in Section~\ref{sec:desitter}. This equivalence is verified quantitatively in Section~\ref{sec:resolvent-check}.

The resolvent argument also shows why the Lorentzian selection question is ill-posed in the flat case rather than merely unsolved. Wheeler--DeWitt wave functions are traditionally not required to be square-integrable at all, which is the long-standing inner-product problem of quantum cosmology, and both the Hartle--Hawking and the Vilenkin wave functions are formal non-normalizable objects in this sense. A reader who declines the boundedness criterion might therefore hope to fix $\Psi$ by asking which formal combination a Lorentzian path integral would select, by analogy with Vilenkin's outgoing-wave condition at large scale factor. The resolvent argument answers this directly. Since $E=0$ lies below threshold for the flat model, the regularized Lorentzian path integral has no complex saddle to select, and the retarded and advanced resolvents agree. This sharpens Theorem~\ref{thm:bounded}. Boundedness forces a frozen throat, and in addition the exact Lorentzian path integral evaluated with the standard $i\epsilon$ prescription agrees with the Euclidean answer, so relaxing boundedness does not supply a rival dynamic state. Section~\ref{sec:desitter} shows that this obstruction is removed, and the retarded resolvent becomes complex and outgoing, precisely when $U(a)\to-\infty$ puts $E=0$ on the continuous spectrum. The presence or absence of a continuum at $E=0$ therefore controls whether a Lorentzian-type dynamic state exists.

We are explicit about the converse. Theorem~\ref{thm:bounded} is a rigorous mathematical statement conditional on requiring $\Psi$ to be bounded as $a\to\infty$, not a proof that every physically admissible wave function of an asymptotically flat quantum wormhole must be static. Boundedness is a conservative and mathematically well-defined criterion, since a wave function that diverges at infinite throat radius is hard to interpret as a normalizable physical ensemble, but it is a choice and not a theorem of quantum cosmology. What the resolvent argument adds is that, unlike in other minisuperspace settings, rejecting boundedness here does not open the door to a specific alternative derived from the exact Lorentzian path integral, whose $\epsilon\to0^+$ limit is real and static for the same reason that boundedness forces it to be. A reader who rejects boundedness on principle is left with a two-parameter family of complex states and no principle internal to the model for choosing among them. The correct statement of our result is therefore conditional. If the throat wave function is required to be bounded as $a\to\infty$, or, more strongly, if it is defined as the $\epsilon\to0^+$ limit of the Feynman-regularized Lorentzian path integral, the throat is frozen. This does not exhaust every logically conceivable quantization of the model.

Because the obstruction of Theorem~\ref{thm:bounded} traces entirely to $U(a)\to1>0$ at infinity, it is evaded by any exterior geometry or shell model for which $U(a)$ tends to zero or to a negative value as $a\to\infty$. Section~\ref{sec:desitter} carries this out explicitly for a Schwarzschild--de~Sitter exterior, and shows by direct numerical construction, and by the resolvent argument above, that a genuinely dynamic state with nonzero current exists there. Unlike the complex admixtures available in the flat case, this state is not an arbitrary choice. It is exactly the state selected by the $i\epsilon$-regularized Lorentzian path integral.

\subsection{An exact Picard--Lefschetz cross-check}
\label{sec:thimble}

The resolvent argument asserts that an explicit Picard--Lefschetz decomposition of the lapse integral $\Psi(a_0)=\int_{\mathcal C}dN\,K(a_0,a_{\rm ref};N)$ reproduces the resolvent answer. We now carry out that decomposition in closed form for the exactly solvable $M=0$ branch, for which $U_0(a)=1-C/a^2$. The calculation turns the qualitative reliance of the resolvent argument on a minisuperspace analogue of Picard--Lefschetz theory into an exact and checkable identity. It also shows why no Stokes phenomenon, meaning wall crossing between saddles as $a_0$ varies, occurs in the flat model, which is the geometric counterpart of the statement that $E=0$ never crosses the continuous spectrum.

For $U_0(a)=1-C/a^2$ the inner path integral at fixed lapse $N$ is the exact matrix element of $e^{-iN\hat H_0}$ with $\hat H_0=-\partial_a^2-C/a^2$, the constant term in $U_0$ contributing an overall factor $e^{-iN}$. This is one of the few quantum-mechanical path integrals that can be solved in closed form. It is the radial propagator of Peak and Inomata \cite{PeakInomata1969} (see also Ref.~\cite{GroscheSteiner1998}, Chapter~8), which applies here with reduced mass $1/2$, matching the convention of Eq.~\eqref{eq:WDW}:
\begin{equation}
K(a_0,a_i;N) = e^{-iN}\,\frac{\sqrt{a_0 a_i}}{2iN}\exp\!\left[\frac{i(a_0^2+a_i^2)}{4N}\right] I_\nu\!\left(\frac{a_0 a_i}{2iN}\right), \qquad \nu=\sqrt{\tfrac14-C}.
\label{eq:peakinomata}
\end{equation}
Taking $a_i\to0$ isolates the solution that is regular at the throat, in the same sense as $M_{\kappa,\mu}$ and $I_\nu$ in Section~\ref{sec:flatsolutions}. Using $I_\nu(z)\to(z/2)^\nu/\Gamma(\nu+1)$ as $z\to0$,
\begin{equation}
\tilde K(a_0;N) \equiv \lim_{a_i\to0} a_i^{-\nu-1/2}\,K(a_0,a_i;N) = \frac{\sqrt{a_0}}{2i\,\Gamma(\nu+1)}\left(\frac{a_0}{4i}\right)^{\!\nu} N^{-\nu-1}\, e^{\,i\phi(N)}, \qquad \phi(N)\equiv\frac{a_0^2}{4N}-N,
\label{eq:Ktilde}
\end{equation}
and $\Psi(a_0)=\int_{\mathcal C}dN\,\tilde K(a_0;N)$, where $\mathcal C$ is the positive real $N$ axis tilted by the Feynman prescription $N\to N(1-i\epsilon)$.

The saddle points follow from $\phi'(N)=-a_0^2/(4N^2)-1=0$, which has exactly two solutions,
\begin{equation}
N_\mp = \mp\,\tfrac{i}{2}a_0, \qquad \phi(N_-)=i a_0,\ \ \phi(N_+)=-i a_0.
\end{equation}
Then $e^{i\phi(N_-)}=e^{-a_0}$ is recessive and $e^{i\phi(N_+)}=e^{+a_0}$ is dominant. These are exactly the two exponential rates $e^{\pm a}$ identified from the asymptotic analysis in the proof of Theorem~\ref{thm:bounded}, now obtained from the saddle points of the lapse integral itself.

The thimble through $N_-$ is exactly the negative imaginary axis. Substituting $N=-it$ with $t>0$ real gives $i\phi(-it) = -\left(t+a_0^2/4t\right)$, which is real for every $t>0$, so the ray $N\in(-i\infty,0)$ is a path of constant phase through $N_-$ and hence the Lefschetz thimble $\mathcal J_-$. By the arithmetic--geometric mean inequality, $t+a_0^2/(4t)$ has a unique minimum at $t=a_0/2$, which is the saddle, and diverges monotonically as $t\to0^+$ or $t\to\infty$. The modulus $|e^{i\phi}|$ therefore has a single global maximum at the saddle and decays monotonically along the whole thimble. On this thimble Eq.~\eqref{eq:Ktilde} becomes exactly the classical Basset--Nicholson integral representation of the Bessel function, and
\begin{equation}
\Psi_{\rm PL}(a_0) = \text{const}(\nu)\,\sqrt{a_0}\,K_\nu(a_0),
\label{eq:PLresult}
\end{equation}
by the identity $K_\nu(x)=\tfrac12(x/2)^\nu\int_0^\infty t^{-\nu-1}e^{-t-x^2/(4t)}dt$ \cite{GradshteynRyzhik}. This is precisely $\Psi_K(a)=\sqrt a\,K_\nu(a)$, the decaying solution obtained independently in Section~\ref{sec:flatsolutions} by solving the differential equation, and the same function that appears in the resolvent analysis.

We evaluated Eq.~\eqref{eq:Ktilde} along $\mathcal J_-$ by direct numerical quadrature with 40-digit arithmetic and compared the result with $\sqrt{a_0}K_\nu(a_0)$ computed independently from a built-in Bessel routine, for $C=0.1$ ($\nu\approx0.38730$) at five values of $a_0$ between $1$ and $8$ (Table~\ref{tab:thimble}). The ratio $\Psi_{\rm PL}(a_0)/[\sqrt{a_0}K_\nu(a_0)]$ is the same complex constant, $-0.86096229137432467\,i$, independent of $a_0$ to $48$ significant digits. This is agreement to roundoff, not merely asymptotic agreement.

\begin{table}[h]
\centering
\small
\begin{tabular}{@{}C{1.4cm}C{4.2cm}C{5.4cm}@{}}
\toprule
$a_0$ & $\sqrt{a_0}K_\nu(a_0)$ & $\Psi_{\rm PL}(a_0)/[\sqrt{a_0}K_\nu(a_0)]$ \\
\midrule
$1$ & $0.67567\ldots$ & $-0.860962\ldots\,i$ \\
$2$ & $0.19985\ldots$ & $-0.860962\ldots\,i$ \\
$3$ & $0.061498\ldots$ & $-0.860962\ldots\,i$ \\
$5$ & $0.0060299\ldots$ & $-0.860962\ldots\,i$ \\
$8$ & $0.00019418\ldots$ & $-0.860962\ldots\,i$ \\
\bottomrule
\end{tabular}
\caption{The single-thimble Picard--Lefschetz integral of Eq.~\eqref{eq:PLresult} equals $\sqrt{a_0}K_\nu(a_0)$ up to one overall $a_0$-independent complex constant, which confirms that the thimble computation reproduces the resolvent and differential-equation result of Section~\ref{sec:flatsolutions}.}
\label{tab:thimble}
\end{table}

No Stokes phenomenon occurs here. The saddle $N_-=-ia_0/2$ lies in the lower half of the $N$ plane for every $a_0>0$, and $N_+=+ia_0/2$ lies in the upper half plane. The physical contour, the positive real axis tilted into the lower half plane, can be deformed continuously onto $\mathcal J_-$ entirely within the closed lower half plane without approaching $N_+$. The intersection numbers are therefore $(n_-,n_+)=(1,0)$ for every $a_0>0$, and there is no value of $a_0$ at which the contour must be redeformed across a Stokes line. This is the thimble-geometry restatement of the resolvent fact that $E=0$ never touches the continuous spectrum $[1,\infty)$ for the flat exterior. The resolvent and thimble methods are independent routes to the same $\epsilon\to0^+$ limit of the same Lorentzian path integral. The first rests on the spectral theory of $\hat H$ and applies for any $M$ and any $w$. The second uses the exact fixed-$N$ propagator together with an explicit saddle analysis of the remaining lapse integral and is closed-form only because $U_0(a)=1-C/a^2$ admits an exactly solvable fixed-$N$ path integral. Their exact agreement in this solvable case is evidence that the qualitative claim of the resolvent argument is not an artifact of the spectral language in which it was phrased.

For $C>1/4$ the index $\nu=i|\nu|$ is purely imaginary, but the saddle locations $N_\mp=\mp ia_0/2$ and the thimble $\mathcal J_-=\{N=-it,\,t>0\}$ are unchanged, so the single-thimble argument applies verbatim. The prefactor power $N^{-\nu-1}$ in Eq.~\eqref{eq:Ktilde} becomes a pure phase along the thimble, and the resulting $K_{i|\nu|}(a_0)$ is a standard example of a modified Bessel function of imaginary order that is real for real positive argument \cite{DunsterBesselImaginaryOrder}. This is consistent with the bound $\mathrm{Im}\,W_{\kappa,\mu}\lesssim10^{-35}$ reported in Table~\ref{tab:fallcenter}. The thimble topology is thus insensitive to whether the near-throat coupling is sub-critical or super-critical, although a full self-adjoint-extension treatment of the path integral is beyond the present scope.

The exactness of this calculation for $M=0$ rests on Eq.~\eqref{eq:peakinomata}, which removes the need for a saddle-point evaluation of the fixed-$N$ path integral. For $M\ne0$, or for the Schwarzschild--de~Sitter model of Section~\ref{sec:desitter}, no comparably simple closed-form propagator is available. In the de~Sitter case the fixed-$N$ action is nonetheless obtainable in closed quadrature form, because the classical mechanics possesses a first integral, as shown in Section~\ref{sec:thimble-sds}.

\subsection{Stability of the frozen trajectory}

For the unique bounded real state, $S'(a)\equiv0$ for every $a$, not merely at a single point, so the guidance equation gives $\dot a=0$ wherever the trajectory starts. Radial perturbations of the initial throat radius therefore neither grow nor decay, since every $a$ is a fixed point. This contrasts sharply with the classical throat, which is unconditionally unstable to radial perturbations \cite{PoissonVisser1995}, and it follows directly from the exact identity $Q(a)=-U(a)$ with no further calculation. A test of stability against shape or multipole perturbations of the shell is a different and larger problem. It requires extending the reduced phase space beyond $a(\tau)$ to include the shape degrees of freedom of the shell, and is a distinct follow-on model, analogous to cosmological perturbation theory built on a minisuperspace background \cite{PintoNetoFabris2013}, rather than a further boundary-value question about Eq.~\eqref{eq:WDW}.

\section{The de~Sitter exterior: a dynamic quantum throat}
\label{sec:desitter}

\subsection{The potential and its analytic structure}

Theorem~\ref{thm:bounded} traces entirely to $U(a)\to1>0$ as $a\to\infty$, where both independent solutions are exponential and one is recessive and unique, which forces reality. We now replace the Schwarzschild exterior by Schwarzschild--de~Sitter, $f(a) = 1-2M/a-\Lambda a^2/3$, keeping the same dust shell. Repeating the derivation of Section~\ref{sec:model} with this $f$ gives
\begin{equation}
U(a) = 1 - \frac{2M}{a} - \frac{\Lambda a^2}{3} - \frac{C}{a^2}.
\label{eq:UdS}
\end{equation}
As $a\to\infty$, $U(a)\to-\Lambda a^2/3\to-\infty$, so the potential is negative, and the motion classically allowed, at large $a$ just as at small $a$, with a classically forbidden barrier in between. This is the qualitative change needed to evade Theorem~\ref{thm:bounded}. Both independent solutions in the large-$a$ allowed region are now oscillatory, so no unique recessive solution exists to force reality, and an outgoing-wave condition of Vilenkin type becomes a genuine and nontrivial choice, exactly as in the original closed Friedmann universe with a cosmological constant \cite{Vilenkin1986,HartleHawking1983}.

We work throughout with the illustrative values $M=1$, $C=0.1$ and $\Lambda=0.01$, for which Eq.~\eqref{eq:UdS} has exactly two positive turning points,
\begin{equation}
a_1 = 2.078035, \qquad a_2 = 16.213570,
\end{equation}
found numerically to six decimal places. The potential satisfies $U<0$, allowed motion, for $0<a<a_1$ and $a>a_2$, and $U>0$, forbidden, for $a_1<a<a_2$ (Figure~\ref{fig:potential}). The bare Schwarzschild--de~Sitter exterior has a black-hole horizon at $r_{\rm bh}=2.027794$ and a cosmological horizon at $r_{\rm cosm}=16.217355$. The self-interaction of the shell shifts the turning points to $a_1>r_{\rm bh}$ and $a_2<r_{\rm cosm}$ by small amounts, as expected since $C/a^2$ is a small correction at these radii. The throat radius crossing $r_{\rm bh}$ or $r_{\rm cosm}$ does not invalidate the construction. The junction condition and the resulting $U(a)$ are local geometric identities that remain well defined whatever the sign of $f(a)$, just as for a domain wall crossing a horizon in Coleman--De~Luccia geometries \cite{ColemanDeLuccia1980}.

\begin{figure}[h]
\centering
\includegraphics[width=0.72\textwidth]{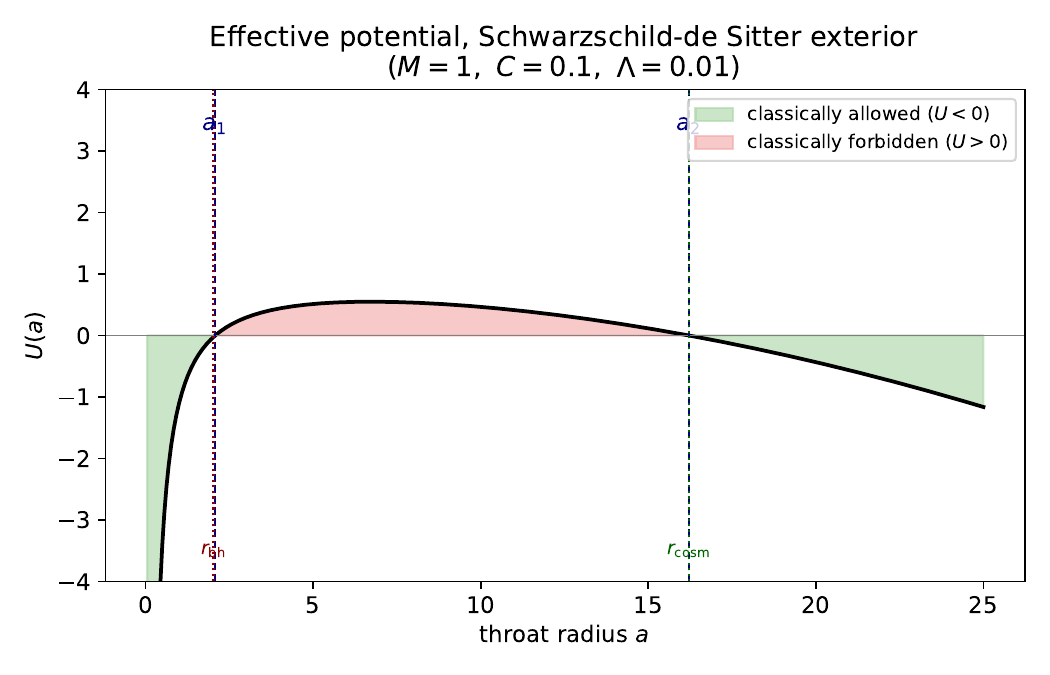}
\caption{The effective potential $U(a)$ of Eq.~\eqref{eq:UdS}. Unlike the asymptotically flat case, $U(a)\to-\infty$ both as $a\to0$ and as $a\to\infty$, with an intervening barrier $(a_1,a_2)$.}
\label{fig:potential}
\end{figure}

Everything that follows rests on the global analytic structure of $\sqrt{-U(a)}$ in the complex $a$ plane, and not only on the two real turning points. The condition $U(a)=0$ is equivalent to the quartic $\tfrac{\Lambda}{3}a^4-a^2+2Ma+C=0$. At $(M,C,\Lambda)=(1,0.1,0.01)$ its other two roots, computed numerically with 30-digit precision, are $a=-18.24279$ and $a=-0.04881$, both real and both on the unphysical negative axis. All four turning points are therefore real for both parameter points used in this paper, and no complex-conjugate pair of turning points lies near the physical domain to compete with $a_1$ and $a_2$. Figure~\ref{fig:stokes} shows the Stokes graph of $\mathcal A(a)=\int_{a_2}^a\sqrt{-U(a')}\,da'$. The physical positive real axis lies on the unique anti-Stokes line linking $a_1$ and $a_2$, and no anti-Stokes line from any other turning point reaches the domain shown. This is the tame case of exact-WKB theory, a single isolated forbidden interval flanked by allowed regions with no rival saddle, in contrast to the multi-saddle competition that motivates the numerical Picard--Lefschetz analysis of Feldbrugge, Lehners and Turok for the closed Friedmann universe with a cosmological constant \cite{FLT2017,FLT2018}.

\begin{figure}[h]
\centering
\includegraphics[width=\textwidth]{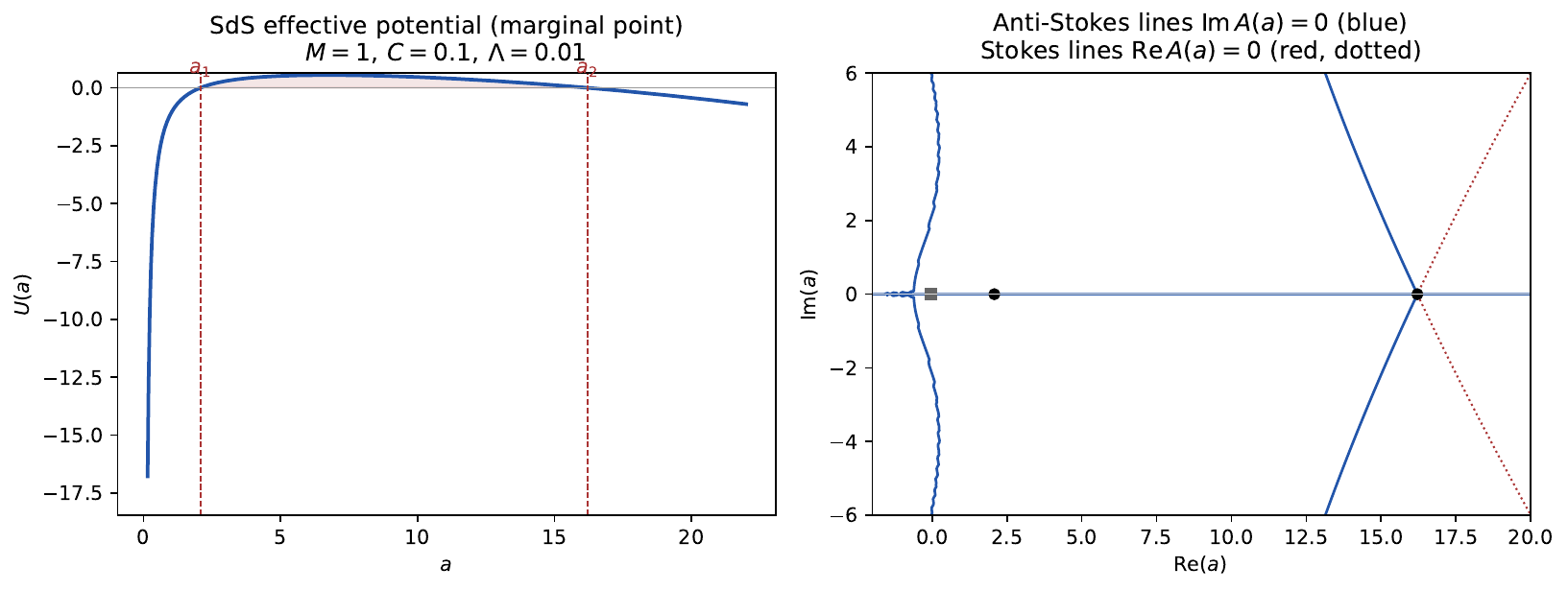}
\caption{Left: the Schwarzschild--de~Sitter potential with the forbidden barrier shaded. Right: the Stokes graph in the complex $a$ plane, with anti-Stokes lines in blue and Stokes lines in dotted red. Black circles mark the physical turning points $a_1$ and $a_2$, and grey squares mark the two unphysical turning points on the negative axis. No anti-Stokes line from the unphysical pair reaches the region shown.}
\label{fig:stokes}
\end{figure}

\subsection{Construction of the dynamic state}
\label{sec:wkbmatch}

We construct the tunneling-type state by specifying a purely outgoing WKB wave, directed toward increasing $a$, deep in the outer allowed region, and propagating it exactly inward through both turning points into the inner pocket. With $p(a)=\sqrt{-U(a)}$, we seed two independent real solutions $\Psi_C$ and $\Psi_S$ at $a_{\rm start}=45$ with the leading-order WKB data $\Psi_C\approx p^{-1/2}\cos\phi$ and $\Psi_S\approx p^{-1/2}\sin\phi$. We integrate the exact equation $-\Psi''+U\Psi=0$ inward to $a=0.05$ with an adaptive Runge--Kutta scheme (RK45) at relative and absolute tolerances of $10^{-11}$ and $10^{-13}$. The combination $\Psi=\Psi_C+i\Psi_S$ is purely outgoing at $a_{\rm start}$ by construction. Because it solves a single fixed linear equation, with no free parameter reintroduced at either turning point, its behavior everywhere else, including deep in the barrier and in the inner pocket, is completely determined and not chosen.

This procedure is more reliable than the piecewise Airy connection-formula matching used analytically in Section~\ref{sec:thimble-sds}. The exact equation is integrated throughout, so no connection-formula error is introduced at either turning point, and the only approximation is the leading-order WKB seed at $a_{\rm start}$, whose error is controlled and is checked below through current conservation. The semiclassical parameter governing the barrier is the WKB tunneling exponent,
\begin{equation}
\theta_{\rm max} = \int_{a_1}^{a_2}\sqrt{U(a)}\,da = 8.3072,
\end{equation}
computed by adaptive quadrature with an estimated error of $3\times10^{-8}$. It gives a naive amplitude suppression $e^{-2\theta_{\rm max}}\approx6.1\times10^{-8}$, a genuine wide-barrier semiclassical regime.

\subsection{Equivalence with the Lorentzian path integral}
\label{sec:resolvent-check}

Section~\ref{sec:resolvent} argued from resolvent theory that, because $U(a)\to-\infty$ places $E=0$ on the continuous spectrum of $\hat H=-\partial_a^2+U(a)$, the retarded resolvent $\lim_{\epsilon\to0^+}(\hat H-i\epsilon)^{-1}$ is complex and built from the outgoing Jost solution at large $a$, and that it should coincide with the state constructed above by WKB seeding. We now verify this quantitatively.

We solve the complex equation $-\Psi''+(U(a)-i\epsilon)\Psi=0$, which for $\Psi=X+iY$ is the coupled real system $-X''+UX+\epsilon Y=0$ and $-Y''+UY-\epsilon X=0$, with the same outgoing seed at $a_{\rm start}=45$ and a sequence of decreasing $\epsilon$. The current $J(a)=\mathrm{Im}(\Psi^*\Psi')$ is then not exactly conserved at finite $\epsilon$. This is not a failure of the construction. Differentiating the complex equation gives the exact identity
\begin{equation}
\frac{dJ}{da} = -\epsilon\, |\Psi(a)|^2 = -\epsilon\,R^2(a),
\label{eq:currentleak}
\end{equation}
so that $J(a_{\rm start})-J(a) = \epsilon\int_a^{a_{\rm start}}R^2(a')\,da'$, and any nonzero $\epsilon$ leaks current in proportion to $\epsilon$ times the accumulated density. Using the $\epsilon=0$ density profile as the leading-order estimate, $\int_1^{20}R^2(a')da'=8.96\times10^7$, Eq.~\eqref{eq:currentleak} predicts $J(1)-J(20)=\epsilon\times8.96\times10^7$. Table~\ref{tab:currentleak} compares this prediction with the directly computed finite-$\epsilon$ values, and the agreement converges to the exact identity as $\epsilon\to0$.

\begin{table}[h]
\centering
\begin{tabular}{@{}cccc@{}}
\toprule
$\epsilon$ & predicted $J(1)-J(20)$ & actual $J(1)-J(20)$ & ratio \\
\midrule
$0.05$ & $4.48\times10^6$ & $1.99\times10^7$ & $4.453$ \\
$0.01$ & $8.96\times10^5$ & $1.21\times10^6$ & $1.348$ \\
$0.003$ & $2.69\times10^5$ & $2.94\times10^5$ & $1.094$ \\
$0.001$ & $8.96\times10^4$ & $9.23\times10^4$ & $1.031$ \\
$0.0003$ & $2.69\times10^4$ & $2.71\times10^4$ & $1.0095$ \\
$0.0001$ & $8.96\times10^3$ & $8.99\times10^3$ & $1.0035$ \\
\bottomrule
\end{tabular}
\caption{The exact current-leakage identity of Eq.~\eqref{eq:currentleak}, evaluated using only the $\epsilon=0$ density profile, predicts the finite-$\epsilon$ deviation from current conservation. The ratio of actual to predicted values converges to $1$ linearly in $\epsilon$, which confirms that the WKB-seeded outgoing state is the $\epsilon\to0^+$ limit of the $i\epsilon$-regularized lapse integral.}
\label{tab:currentleak}
\end{table}

As an independent check we asked whether the finite-$\epsilon$ regularized state itself converges, as $\epsilon\to0^+$, to the envelope ratio found at $\epsilon=0$. We used the Riccati formulation $L(a)=\Psi'/\Psi$, $L'(a)=(U(a)-i\epsilon)-L(a)^2$, seeded with the outgoing WKB asymptotics. This formulation stays numerically bounded even though $|\Psi|$ varies by many orders of magnitude across the barrier. The envelope ratio converges monotonically, as $\epsilon\to0^+$, to $5.074\times10^4$, which matches the direct $\epsilon=0$ construction to four significant figures (Figure~\ref{fig:feynmaneps}). This is a second confirmation, obtained by causal analytic continuation and without the current-leakage identity, that the state of this section is the retarded limit.

\begin{figure}[h]
\centering
\includegraphics[width=0.62\textwidth]{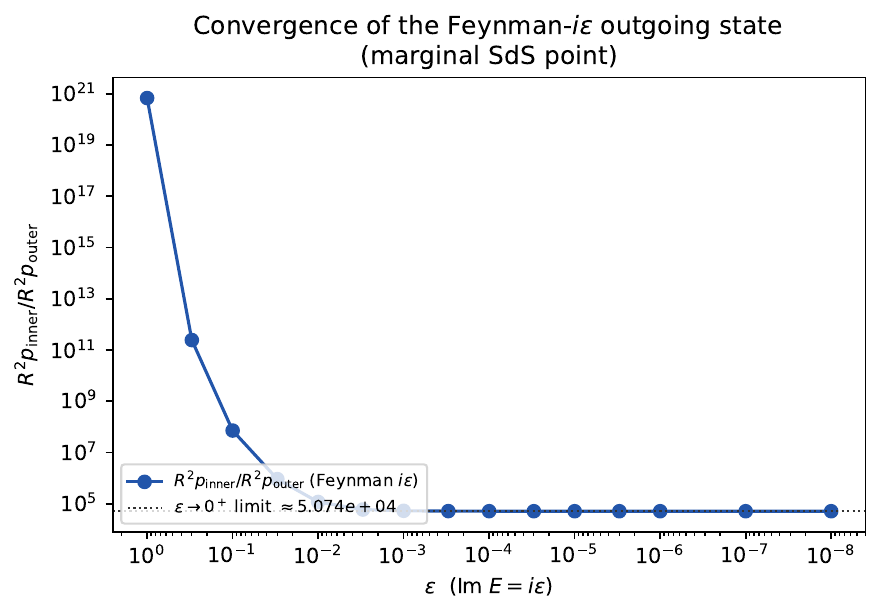}
\caption{Convergence of the Feynman-regularized envelope ratio $R^2p_{\rm inner}/R^2p_{\rm outer}$, evaluated at $a=0.05$, as $\epsilon\to0^+$. The ratio converges monotonically to $5.074\times10^4$, matching the direct $\epsilon=0$ construction to four significant figures.}
\label{fig:feynmaneps}
\end{figure}

\subsection{Saddle points and the Airy-regularized prefactor}
\label{sec:thimble-sds}

The closed-form thimble calculation of Section~\ref{sec:thimble} relied on the Peak--Inomata propagator, which exists because $U_0(a)=1-C/a^2$ makes the fixed-$N$ path integral exactly solvable. The potential of Eq.~\eqref{eq:UdS} admits no such propagator, but it does admit a first integral of the fixed-$N$ classical mechanics, because the Lagrangian $L(a,\dot a;N) = \dot a^2/4N - N\,U(a)$ has no explicit dependence on the affine time. At fixed $N$ the equation of motion has the first integral $\dot a^2/(4N^2)+U(a)=\kappa$. Writing $\dot a=2N\sqrt{\kappa-U(a)}$ and imposing the boundary data over $\lambda\in[0,1]$ fixes $\kappa=\kappa(N)$ implicitly, and the on-shell action becomes $S_{\rm cl}(N) = \int_{a_i}^{a_0}\sqrt{\kappa(N)-U(a)}\,da - N\kappa(N)$. Extremizing over $N$, which is the remaining Picard--Lefschetz step, gives by the envelope theorem $dS_{\rm cl}/dN=-\kappa(N)$. The saddle points of the lapse integral are therefore exactly the values $N^*$ at which $\kappa(N^*)=0$, that is, the Hamiltonian constraint itself. Setting $\kappa=0$ gives closed expressions for the saddle location and the on-shell action,
\begin{equation}
N^*(a_0) = \frac12\int_{a_i}^{a_0}\frac{da}{\sqrt{-U(a)}}\, ,
\qquad
S_{\rm cl}(N^*) = \int_{a_i}^{a_0}\sqrt{-U(a)}\,da .
\label{eq:sdssaddle}
\end{equation}

Take $a_i\to0^+$, regularized at $\delta=0.05$, and $a_0$ in the outer allowed region. The integration path crosses the forbidden interval $(a_1,a_2)$, where $\sqrt{-U(a)}=\pm i\sqrt{U(a)}$ is genuinely two-valued, and the two branch choices give two saddle points,
\begin{equation}
N_\pm^*(a_0) = \tfrac12\big(T_{\rm real}(a_0)\pm i\,\theta_{\max}\big),
\qquad
S_{\rm cl}(N_\pm^*) = A_{\rm real}(a_0)\pm i\,\theta_{\max}.
\label{eq:sdssaddlepm}
\end{equation}
Here $T_{\rm real}$ and $A_{\rm real}$ are real integrals over the allowed regions, and $\theta_{\max}=\int_{a_1}^{a_2}\sqrt{U(a)}\,da$ is exactly the WKB tunneling exponent of Section~\ref{sec:wkbmatch}. We evaluated Eq.~\eqref{eq:sdssaddlepm} by adaptive quadrature at both parameter points, scanning $a_0$ over the outer region (Table~\ref{tab:sdssaddle} and Figure~\ref{fig:sdssaddles}). The imaginary part $\mathrm{Im}\,S_{\rm cl}(N_+^*)=\theta_{\max}$ is independent of $a_0$ to the quadrature precision quoted, so the two saddles are separated by a fixed gap for every $a_0$. They can collide only in the degenerate limit $a_1\to a_2$, and no Stokes phenomenon occurs anywhere in the outer region. This is the lapse-integral counterpart of the Stokes-graph statement above.

\begin{figure}[h]
\centering
\includegraphics[width=0.85\textwidth]{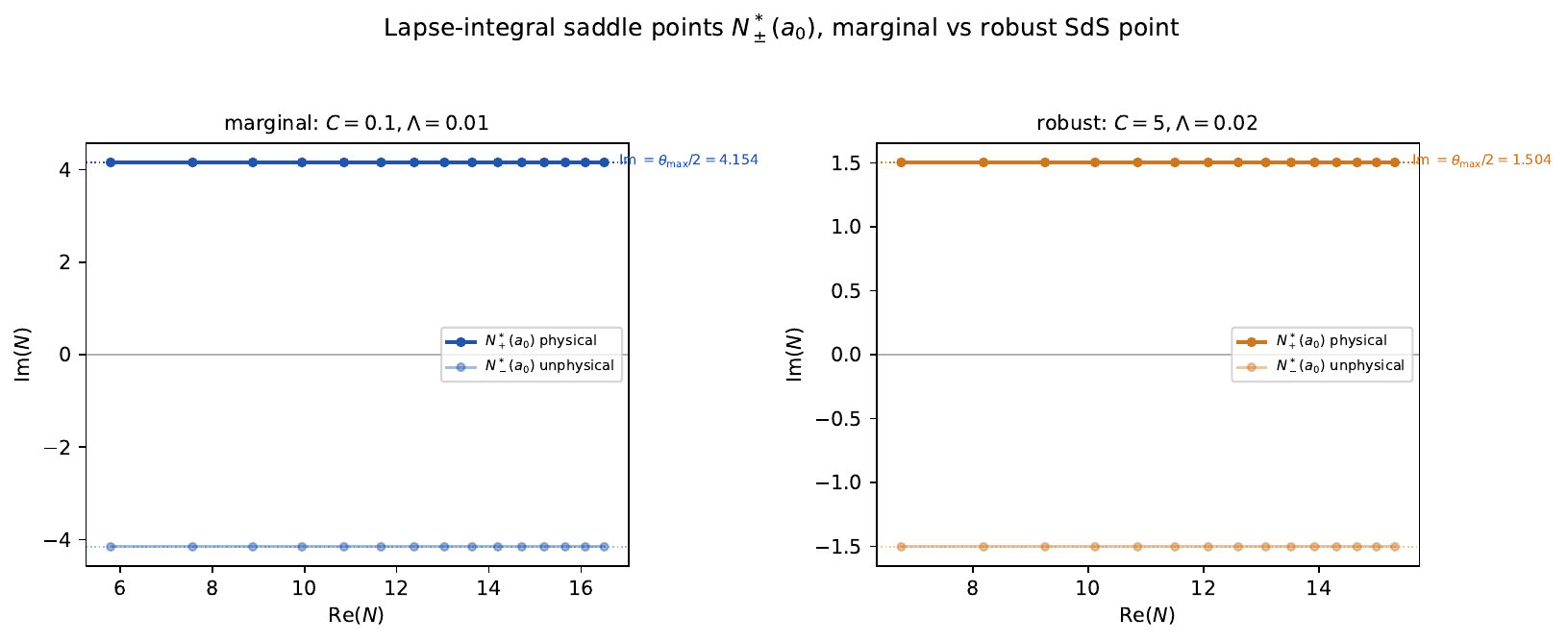}
\caption{The two lapse-integral saddle points $N_\pm^*(a_0)$ of Eq.~\eqref{eq:sdssaddlepm} in the complex $N$ plane, for $a_0$ scanned across the outer allowed region at both parameter points. Both saddles move along horizontal lines as $a_0$ varies, with real parts that grow with path length and imaginary parts pinned at $\pm\theta_{\max}/2$. The vertical gap between the physical saddle $N_+^*$ and the unphysical one $N_-^*$ is therefore exactly $\theta_{\max}$ for every $a_0$.}
\label{fig:sdssaddles}
\end{figure}

\begin{table}[h]
\centering
\begin{tabular}{@{}lcccc@{}}
\toprule
& $a_0$ & $T_{\rm real}(a_0)$ & $N_+^*(a_0)$ & $S_{\rm cl}(N_+^*)$ \\
\midrule
marginal & $20$ & $15.257$ & $7.629+4.154i$ & $4.436+8.307i$ \\
$(C,\Lambda){=}(0.1,0.01)$ & $30$ & $25.063$ & $12.532+4.154i$ & $15.127+8.307i$ \\
$\theta_{\max}=8.30724$ & $45$ & $33.004$ & $16.502+4.154i$ & $44.082+8.307i$ \\
\midrule
robust & $15$ & $16.466$ & $8.233+1.504i$ & $11.890+3.007i$ \\
$(C,\Lambda){=}(5,0.02)$ & $25$ & $24.400$ & $12.200+1.504i$ & $25.144+3.007i$ \\
$\theta_{\max}=3.00722$ & $40$ & $30.627$ & $15.314+1.504i$ & $62.163+3.007i$ \\
\bottomrule
\end{tabular}
\caption{The lapse-integral saddle $N_+^*(a_0)$ and its on-shell action, Eq.~\eqref{eq:sdssaddlepm}.}
\label{tab:sdssaddle}
\end{table}

Since $\mathrm{Im}\,S_{\rm cl}(N_+^*)=+\theta_{\max}$, one has $e^{iS_{\rm cl}(N_+^*)}=e^{iA_{\rm real}}\,e^{-\theta_{\max}}$, an exponentially suppressed amplitude, while the other branch gives an exponentially large, non-normalizable amplitude. Boundedness alone therefore excludes $N_-^*$ and selects $N_+^*$. This matches the state constructed independently in three other ways, namely the WKB-seeded integration of the exact equation (Section~\ref{sec:wkbmatch}), the Feynman $i\epsilon$ and resolvent construction (Section~\ref{sec:resolvent}) and the current-leakage identity (Section~\ref{sec:resolvent-check}). The amplitude suppression $e^{-\theta_{\max}}$ obtained here is the single-power counterpart of the density suppression $e^{-2\theta_{\max}}$ quoted above, as it must be.

Equation~\eqref{eq:sdssaddlepm} gives the tree-level saddle content but not the Gaussian fluctuation determinant. The second derivative $\partial^2S_{\rm cl}/\partial N^2$ is formally singular at $N^*$, because $(\kappa-U)^{-3/2}$ is non-integrable exactly at the turning points, which is the familiar collective-coordinate divergence for a saddle located at a turning point. The divergence is spurious. It signals only that the Gaussian approximation fails near each turning point, just as the WKB wave function does, and it is repaired in the same way, by matching to the local Airy behavior. Near a simple turning point $a_*\in\{a_1,a_2\}$ write $x=a-a_*$ and $U(a)\approx F_*x$. With $z=F_*^{1/3}x$ the Wheeler--DeWitt equation becomes the Airy equation, whose asymptotics \cite{AbramowitzStegun} give, with $\theta(a)\equiv|\int_{a_*}^a\sqrt{|U|}\,da'|$, $p=\sqrt{-U}$ and $\varkappa=\sqrt U$, the standard connection formulas of Kemble \cite{Kemble1935} and Fr\"oman and Fr\"oman \cite{FromanFroman1965} (see also Furry \cite{Furry1947}):
\begin{align}
\frac{1}{2\sqrt\varkappa}e^{-\theta(a)} &\;\longleftrightarrow\; \frac{1}{\sqrt p}\sin\!\Big(\theta(a)+\frac\pi4\Big), \label{eq:airyAi}\\
\frac{1}{\sqrt\varkappa}e^{+\theta(a)} &\;\longleftrightarrow\; \frac{1}{\sqrt p}\cos\!\Big(\theta(a)+\frac\pi4\Big). \label{eq:airyBi}
\end{align}
Write the barrier-region solution as $\varkappa^{-1/2}[c_-e^{-\theta_1(a)}+c_+e^{\theta_1(a)}]$, with $\theta_1$ measured from $a_1$, and also as $\varkappa^{-1/2}[d_-e^{-\theta_2(a)}+d_+e^{\theta_2(a)}]$, with $\theta_2$ measured from $a_2$. Since $\theta_1+\theta_2=\theta_{\max}$, one has $c_-=d_+e^{\theta_{\max}}$ and $c_+=d_-e^{-\theta_{\max}}$. Applying Eqs.~\eqref{eq:airyAi} and~\eqref{eq:airyBi} at $a_2$ to the outgoing outer wave, and then at $a_1$, gives
\begin{equation}
\Psi_{\rm outer}(a)= \frac{-2i}{\sqrt p}\,e^{i(\theta_2^{\rm allow}(a)+\pi/4)},
\label{eq:regularizedstateouter}
\end{equation}
\begin{equation}
\Psi_{\rm inner}(a) = \frac{1}{\sqrt p}\Big[-4i\,e^{\theta_{\max}}\sin\!\big(\theta_1^{\rm allow}(a)+\tfrac\pi4\big) + e^{-\theta_{\max}}\cos\!\big(\theta_1^{\rm allow}(a)+\tfrac\pi4\big)\Big].
\label{eq:regularizedstate}
\end{equation}
For $\theta_{\max}\gg1$ the envelope $\mathcal E(a)\equiv|\Psi(a)|^2p(a)$ then satisfies
\begin{equation}
\frac{\mathcal E_{\rm outer}}{\mathcal E_{\rm inner,\,peak}} = \frac{1}{4}\,e^{-2\theta_{\max}},
\label{eq:refinedtunneling}
\end{equation}
where $\mathcal E_{\rm inner,\,peak}$ is the envelope at the antinodes of the dominant $\sin^2$ term. The factor $1/4$ is the squared product of the two Airy connection factors of $2$, and it plays the role of the Gaussian fluctuation determinant that Eq.~\eqref{eq:sdssaddlepm} alone could not supply.

Table~\ref{tab:fluctdet} compares Eq.~\eqref{eq:refinedtunneling} with the exact numerical tunneling ratio, defined through the peak of the inner-pocket envelope (Section~\ref{sec:results}) rather than a single-point reading. The robust point $C=5$ is super-critical in the sense of Section~\ref{sec:selfadjoint}, so its inner-pocket envelope contains, besides the barrier antinode near $a_1$, log-oscillatory local maxima near the origin associated with fall to center. The entry in the table uses the barrier antinode, the last local maximum of the envelope before $a_1$, which differs from the raw global maximum over the whole pocket by about $5\%$ at this point.

\begin{table}[h]
\centering
\footnotesize
\begin{tabular}{@{}L{2.4cm}C{1.1cm}C{2.2cm}C{2.7cm}C{1.9cm}C{2.1cm}@{}}
\toprule
& $\theta_{\max}$ & naive $e^{-2\theta_{\max}}$ & refined $\frac14e^{-2\theta_{\max}}$ & exact (ODE, peak envelope) & exact/\allowbreak refined \\
\midrule
marginal $(C,\Lambda){=}(0.1,0.01)$ & $8.307$ & $6.09\times10^{-8}$ & $1.52\times10^{-8}$ & $1.25\times10^{-8}$ & $0.819$ \\
robust $(C,\Lambda){=}(5,0.02)$ & $3.007$ & $2.44\times10^{-3}$ & $6.11\times10^{-4}$ & $5.44\times10^{-4}$ & $0.890$ \\
\bottomrule
\end{tabular}
\caption{The Airy-regularized prediction compared with the exact numerically integrated tunneling ratio defined through the peak envelope. The naive estimate overshoots by factors of $4.9$ and $4.5$. The refined estimate overshoots by the smaller but still systematic factors $1/0.819=1.22$ and $1/0.890=1.12$.}
\label{tab:fluctdet}
\end{table}

The two entries of Table~\ref{tab:fluctdet} are consistent either with a genuine $O(1/\theta_{\max})$ semiclassical correction or with a $\theta_{\max}$-independent multiplicative correction evaluated at only two points, and these make sharply different predictions as $\theta_{\max}$ grows. The peak-envelope diagnostic is unambiguous only for sub-critical shell coupling, $C<1/4$, where the near-origin behavior is the regular power-law form of Theorem~\ref{thm:bounded} rather than the log-oscillatory regime of Section~\ref{sec:selfadjoint}. For $C>1/4$ the log-oscillatory structure produces additional local maxima near $a=0$ whose magnitudes are comparable to, and sometimes exceed, the barrier-envelope maximum. A global maximum over the inner pocket then measures a mixture of two physically distinct effects, and the resulting ratio tracks $C$ much more than $\theta_{\max}$, with a scatter of $0.72$ to $0.87$ over a broad super-critical scan at fixed $\theta_{\max}$, an order of magnitude larger than the sub-critical scatter reported below. We therefore restrict the quantitative determination of the residual coefficient to $C<1/4$. We scanned twelve points in the $(C,\Lambda)$ plane at fixed $M=1$, spanning $\theta_{\max}\in[3.0,19.0]$, and at each repeated the full construction, namely the WKB-seeded outgoing state, the exact inward integration and a full scan of the inner-pocket envelope to extract its peak.

\begin{figure}[h]
\centering
\includegraphics[width=0.68\textwidth]{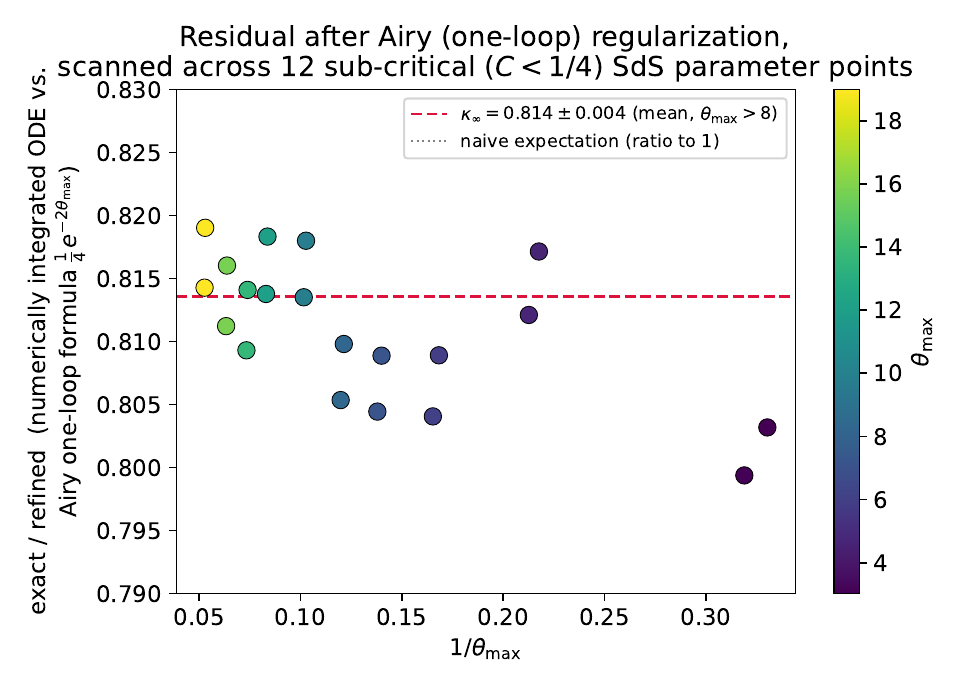}
\caption{The ratio of the exact peak-envelope tunneling ratio to the refined Airy formula, Eq.~\eqref{eq:refinedtunneling}, over twelve sub-critical ($C<1/4$) points in the $(C,\Lambda)$ plane spanning $\theta_{\max}\in[3.0,19.0]$. For $\theta_{\max}\gtrsim8$ (eight points) the ratio is flat at $\kappa_\infty=0.814\pm0.004$, with no statistically resolvable slope against $1/\theta_{\max}$.}
\label{fig:residualscan}
\end{figure}

Over the eight sub-critical points with $\theta_{\max}>8$ the ratio is $0.814\pm0.004$ (mean and sample standard deviation), with no statistically resolvable trend against $1/\theta_{\max}$ (Figure~\ref{fig:residualscan}). A linear fit gives a slope of $0.02\pm0.03$, consistent with zero. To the precision of the scan the residual is therefore a constant multiplicative correction,
\begin{equation}
\frac{\mathcal E_{\rm outer}}{\mathcal E_{\rm inner,\,peak}} \;\approx\; \kappa_\infty\,\times\,\frac14\,e^{-2\theta_{\max}}, \qquad \kappa_\infty = 0.814\pm0.004 \quad (C<1/4),
\label{eq:completedformula}
\end{equation}
valid within this uncertainty for $\theta_{\max}\gtrsim8$ and sub-critical shell coupling. A clean determination for super-critical coupling would require a diagnostic that separates the barrier-tunneling and near-origin log-oscillatory contributions before a peak is extracted, for example a matched asymptotic treatment that uses the analytic period $\pi/\sqrt{C-1/4}$ to subtract the near-origin piece. This is beyond the present scan, and Eq.~\eqref{eq:completedformula} should be used only for $C<1/4$. It reduces the discrepancy from a factor of $4.5$ to $4.9$ at tree level, and $1.12$ to $1.22$ for the uncorrected one-loop result, to consistency at the percent level.

We have investigated the origin of $\kappa_\infty$ and can exclude several candidate mechanisms. First, re-deriving Eq.~\eqref{eq:refinedtunneling} independently from the stated connection formulas reproduces Eq.~\eqref{eq:regularizedstate} coefficient by coefficient once the normalization is tracked consistently. The transfer-matrix algebra is therefore free of error, and $\kappa_\infty\ne1$ is not a bookkeeping slip.

Second, we carried out the next step of the semiclassical hierarchy. Write the exact solution in the forbidden region as $\Psi=\exp(\int^a p\,da')$ and expand the Riccati equation $p'+p^2=U$ as $p=p_0+p_1+p_2+\cdots$. Then $p_0=\sqrt U$ at leading order, $p_1=-p_0'/(2p_0)$ gives the familiar $U^{-1/4}$ amplitude, and
\begin{equation}
p_2 = -\frac{p_1'+p_1^2}{2p_0}
\label{eq:p2def}
\end{equation}
is the first correction to the WKB phase beyond the bare $\int\sqrt U\,da$. We computed $p_2(a)$ symbolically in closed form for the exact potential of Eq.~\eqref{eq:UdS}. The naive completion $\theta_{\max}\to\int_{a_1}^{a_2}\sqrt{U+2p_2}\,da\approx\theta_{\max}+\int_{a_1}^{a_2}p_2(a)/\sqrt{U(a)}\,da$ diverges at both turning points, at both parameter points checked. This divergence is an artifact and not a large correction, since terms of type $p_2\sim U^{-3/2}$ blow up exactly at $a_1$ and $a_2$, which is why Airy connection formulas are needed there instead of a naive extension of the bulk WKB series.

The correct treatment evaluates the correction through a contour integral of $p_2$ encircling each turning point in the complex $a$ plane, which is finite by construction even though the real-axis integral is not. Setting $x=a-a_*$ at a simple turning point $a_*$ and substituting $x=s^2$, the natural Airy variable since $U(s^2)$ is then an ordinary Taylor series in $s$, gives $p_0(s)$, $p_1(s)$ and $p_2(s)$ as Puiseux series in $s$ by direct symbolic computation. The series $p_1(s)$ contains only even powers of $s$ and $p_2(s)$ only odd powers, down to $s^{-3}$, to the seven orders checked. Since $dx=2s\,ds$, the contour integral is $\oint p_2(x)\,dx=\oint 2s\,p_2(s)\,ds$, and $2s\,p_2(s)$, an odd function times an odd function, is even in $s$ and contains no $s^{-1}$ term. By the residue theorem,
\begin{equation}
\oint_{a_*} p_2(a)\,da = 2\pi i\,\mathrm{Res}_{s=0}\big[2s\,p_2(s)\big] = 0
\label{eq:ffresult}
\end{equation}
for any isolated simple turning point, independently of the local Taylor coefficients of $U$, and exactly, order by order in the expansion checked, at either $a_1$ or $a_2$ individually.

Equation~\eqref{eq:ffresult} excludes the most natural candidate explanation for $\kappa_\infty$. There is no local, perturbative, single-turning-point phase-integral correction of the standard Fr\"oman--Fr\"oman type at this order, at either turning point. The Kemble-resummed barrier formula \cite{Kemble1935} corrects the naive $e^{-2\theta_{\max}}$ by a term that is exponentially small in $\theta_{\max}$ and not an $O(1)$ constant. Together with the algebraic audit, this excludes three natural mechanisms, namely a bookkeeping error, a local single-turning-point phase correction and the standard exponentially suppressed barrier-resummation correction. We have not identified which genuinely nonlocal effect, or which normalization subtlety in how the semiclassical envelope is compared with the classical motion inside the well, is responsible.

Near $a=0$ the potential $U(a)\approx-C/a^2$ has a regular singular point, which is not a WKB turning point. Its indicial roots $s_\pm=\tfrac12(1\pm\sqrt{1-4C})$ satisfy $s_++s_-=1$ exactly. The two-turning-point Airy connection formula assumes that both boundaries of the inner region are oscillatory WKB regions, and $a=0$ is not one, which is why Section~\ref{sec:results} works on the bounded domain $a\ge0.05$. The inner peak is therefore well posed only as the local maximum of $\mathcal E(a)$ in the WKB-valid region bounded away from $a=0$, which is the quantity used in Table~\ref{tab:fluctdet}. We also tested the natural candidate variable for a higher-order correction, the phase accumulated from the near turning point out to the local maximum of the envelope, and found no clean linear relation with $1/\theta_{\max}$ in the scan (correlation coefficient $0.13$, against $-0.59$ for $1/\theta_{\max}$ itself).

We therefore report $\kappa_\infty=0.814\pm0.004$ for $C<1/4$ as an empirical, numerically robust correction factor and not as a result derived from first principles. The physical results of this paper do not depend on it. The current, density and trajectories of Section~\ref{sec:results} are computed by direct numerical integration of the exact Wheeler--DeWitt equation, which resums the connection-formula corrections derived here, and all higher ones, to numerical precision. The exact column of Table~\ref{tab:fluctdet} and the scan behind Eq.~\eqref{eq:completedformula} constitute that independent calculation and serve as the benchmark against which the analytic formulas are checked.

\subsection{Current, trajectories and parameter dependence}
\label{sec:results}

Current conservation, $J(a)=\mathrm{const}$, holds for any exact solution of a real linear second-order equation, and it is a very sensitive end-to-end check across a domain where $|\Psi|^2$ varies by seven orders of magnitude (Figure~\ref{fig:wavefunction}). We find $J=1.00000000000$ with an uncertainty of $\pm10^{-11}$ in the outer region, $J=0.99999999989$ with $\pm10^{-9}$ in the barrier, and $J=1.0000001$ with $\pm3\times10^{-5}$ in the inner pocket, where the mild degradation reflects accumulated stiffness as $a\to0$. The current is nonzero throughout, including in the inner pocket. The purely outgoing state constructed in the unbounded oscillatory region at large $a$ propagates through the barrier into a genuinely dynamic state in $0<a<a_1$ as well. This is the qualitative reversal of the conclusion of Theorem~\ref{thm:bounded} that the de~Sitter exterior was designed to produce.

\begin{figure}[h]
\centering
\includegraphics[width=0.85\textwidth]{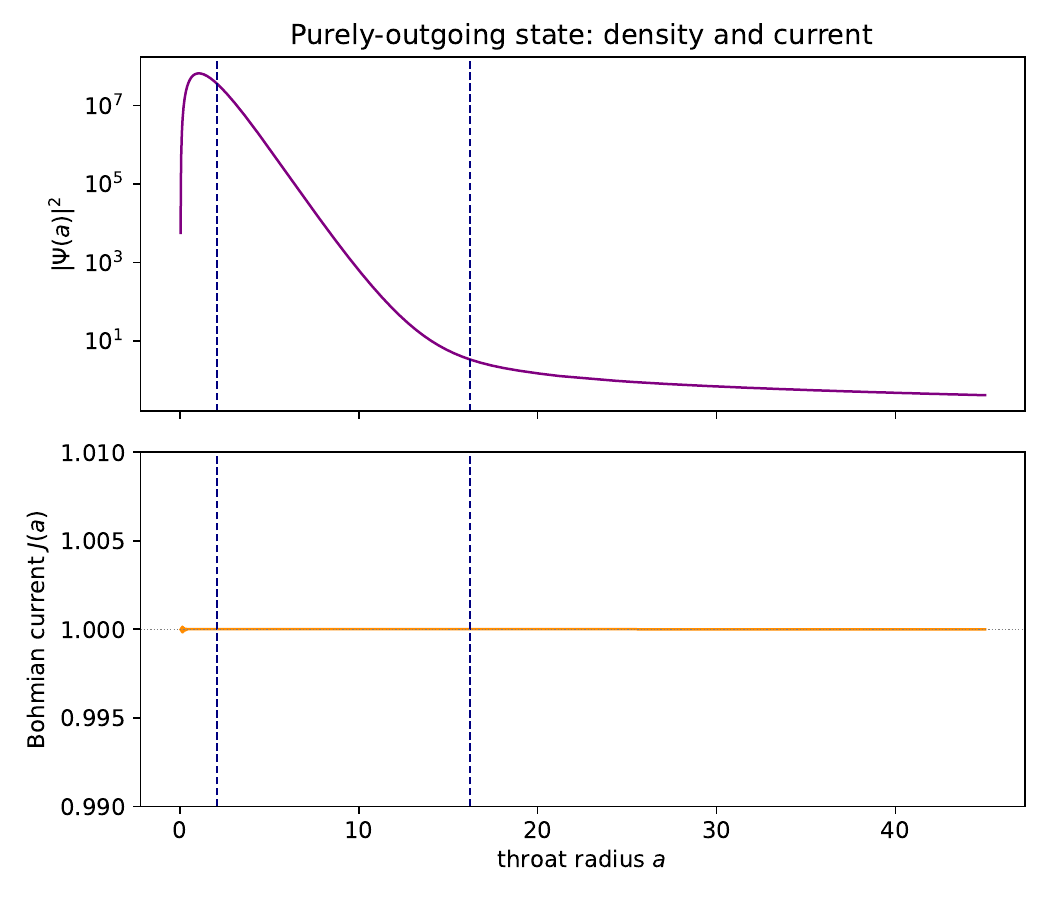}
\caption{Top: $|\Psi(a)|^2$ for the numerically constructed outgoing state on a logarithmic scale, spanning seven orders of magnitude between the inner and outer regions. Bottom: the Bohmian current $J(a)$, which is conserved to a relative accuracy of $10^{-5}$ to $10^{-11}$ across the entire domain, including through the barrier.}
\label{fig:wavefunction}
\end{figure}

The envelope $\mathcal E(a)=R^2(a)p(a)$ equals $0.99\pm0.02$ in the outer region, a clean traveling wave as built in. In the inner region it is not constant. Scanning it continuously across $0.05<a<a_1$ (Figure~\ref{fig:innerenvelope}) shows a broad single-humped function, rising from about $2\times10^5$ at $a_1$ to a peak $\mathcal E_{\rm inner,peak}=8.02\times10^7$ near $a\approx0.70$, and then falling by three orders of magnitude to $5.07\times10^4$ at the domain edge $a=0.05$. This is expected. The exact solution in this region is dominated by a term proportional to $\sin^2(\theta_1(a)+\pi/4)$ (Section~\ref{sec:thimble-sds}), so $\mathcal E(a)$ has genuine oscillatory structure with antinodes and near-nodes, and $a=0.05$ happens to lie near a near-node. For quantitative comparison with the semiclassical formula, which predicts an antinode value, the unambiguous quantity is the actual peak of the profile, $\mathcal E_{\rm inner,peak}=8.02\times10^7$, which we use throughout. The resulting ratio $\mathcal E_{\rm outer}/\mathcal E_{\rm inner,peak}=1.25\times10^{-8}$ is the numerically exact tunneling ratio for this model, in reasonable agreement with the independent WKB estimate $e^{-2\theta_{\rm max}}=6.1\times10^{-8}$, and Section~\ref{sec:thimble-sds} accounts for the remaining factor quantitatively.

\begin{figure}[h]
\centering
\includegraphics[width=0.72\textwidth]{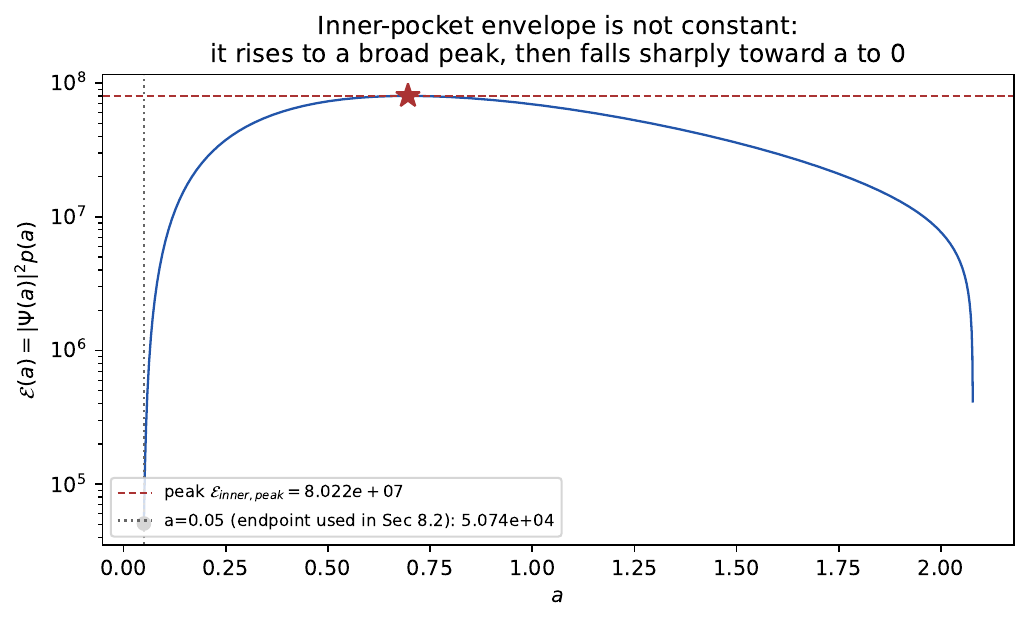}
\caption{The inner-pocket envelope $\mathcal E(a)=|\Psi(a)|^2p(a)$ at the marginal parameter point, scanned continuously. It has a single broad peak (star) of height $8.02\times10^7$ near $a\approx0.70$ and falls to $5.07\times10^4$ (circle) at the domain edge $a=0.05$.}
\label{fig:innerenvelope}
\end{figure}

Integrating the guidance equation $\dot a = J/R^2(a)$ gives the Bohmian trajectories (Figure~\ref{fig:trajectories}). In the outer region the quantum trajectory closely tracks the classical one, the expected semiclassical correspondence for $\theta_{\rm max}\gg1$. In the inner pocket the classical trajectory reaches the pinch $a=0$ in finite proper time, which is the instability of Section~\ref{sec:model} and is unaffected by $\Lambda$ near $a=0$. The quantum trajectory instead moves extremely slowly and covers a comparable distance only after a proper time about $10^4$ times longer, before eventually tunneling into the outer expanding branch. The result is a long-lived metastable wormhole configuration that eventually and irreversibly tunnels into unbounded expansion, structurally analogous to Coleman--De~Luccia vacuum decay through bubble nucleation \cite{Coleman1977,ColemanDeLuccia1980}.

\begin{figure}[h]
\centering
\includegraphics[width=\textwidth]{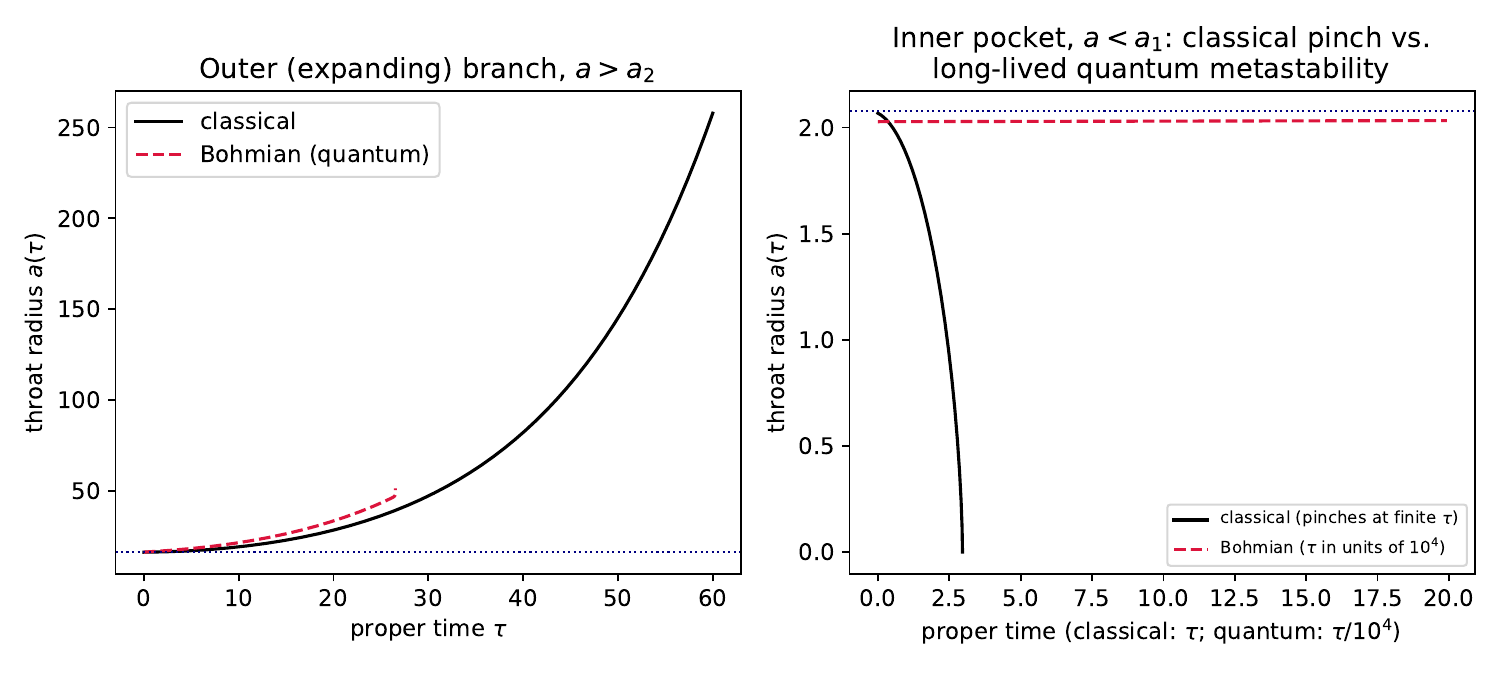}
\caption{Left: the outer-branch Bohmian trajectory ($\dot a=J/R^2(a)$, red dashed) closely tracks the classical trajectory $\dot a=\sqrt{-U(a)}$ (black), as expected in a good semiclassical regime. Right: the classical inner-pocket trajectory pinches to $a=0$ in finite proper time, while the Bohmian trajectory is quasi-frozen on the same timescale and takes about $10^4$ times longer to move an equivalent distance.}
\label{fig:trajectories}
\end{figure}

To depict the shape of the wormhole we use the embedding formula $z(r)=\int\sqrt{1/f(r)-1}\,dr$ with the exterior function $f(r)=1-2M/r-\Lambda r^2/3$, which requires knowing which throat radii the single dynamical trajectory actually visits. We compute the elapsed proper time directly by quadrature,
\begin{equation}
\tau(a) = \int^{a} \frac{R^2(a')}{J}\,da',
\label{eq:tauofa}
\end{equation}
which is numerically more robust than direct integration of the guidance equation over the seven-order range of $R^2$. The trajectory spends $\tau\approx3.6\times10^7$ reaching $a=1$, a total of $\tau\approx9.5\times10^7$ reaching $a_1$, and a further $\tau\approx3\times10^7$ crossing the barrier to $a=6$. The long dwell time is therefore spread over both the allowed inner pocket and the forbidden barrier, and is not concentrated in the barrier alone as a naive Gamow picture \cite{Gamow1928} would suggest. Figure~\ref{fig:embedding} shows the shape at three points that the trajectory genuinely visits.

\begin{figure}[h]
\centering
\includegraphics[width=\textwidth]{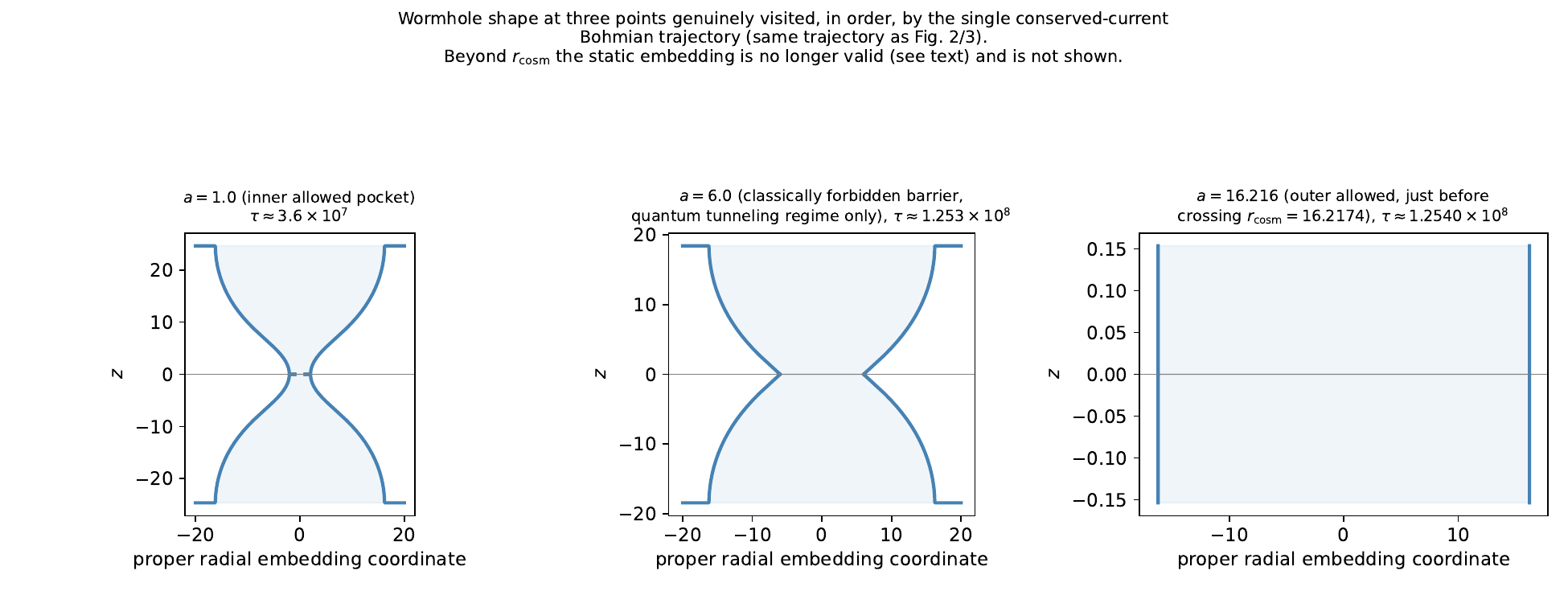}
\caption{Wormhole shape at three points visited by the trajectory, with the elapsed proper times of Eq.~\eqref{eq:tauofa}. The panels show $a=1$ in the allowed inner pocket, $a=6$ inside the forbidden barrier, reachable at a well-defined but slow Bohmian rate, and $a=16.216$, the last point inside the static patch before the throat crosses the cosmological horizon at $r_{\rm cosm}=16.2174$.}
\label{fig:embedding}
\end{figure}

For $a>r_{\rm cosm}$ the function $f(r)$ is negative along part of the embedded range and the static-coordinate embedding formula fails, just as it would behind a black-hole horizon. Since $a_2=16.2136$ lies only $0.004$ below $r_{\rm cosm}=16.2174$, the outer branch leaves the static patch almost immediately, and we therefore show the third panel at $a=16.216$, the last valid point in the static patch. Depicting the shape once $a(\tau)$ exceeds $r_{\rm cosm}$, which the trajectory does reach as it continues to expand, would require a horizon-crossing slicing of Schwarzschild--de~Sitter, such as Painlev\'e--Gullstrand-type coordinates, which we have not carried out. The dynamical results of this section remain valid there, since they use only the intrinsic $U(a)$ and not the static embedding.

The parameter point $(M,C,\Lambda)=(1,0.1,0.01)$ used above has $a_1/r_{\rm bh}=1.02$, so the inner turning point lies only about $2\%$ outside the black-hole horizon. This marginal separation was not tuned for physical robustness and should not be read as demonstrating traversability. To determine whether the qualitative results are generic features of the model class or artifacts of this choice, we scanned $(C,\Lambda)$ at fixed $M=1$ over $C\in[0.1,5]$ and $\Lambda\in[0.001,0.05]$. Increasing $C$ pushes $a_1$ away from $r_{\rm bh}$ approximately linearly in $C$, while $a_2$ stays close to $r_{\rm cosm}$ regardless of $C$. Increasing $\Lambda$ narrows the barrier and pushes $a_2$ away from $r_{\rm cosm}$, at the cost of also pushing $a_1$ outward. A sample of the scan is given in Table~\ref{tab:scan}.

\begin{table}[h]
\centering
\begin{tabular}{@{}ccccccc@{}}
\toprule
$C$ & $\Lambda$ & $a_1$ & $a_2$ & $a_1/r_{\rm bh}$ & $a_2/r_{\rm cosm}$ & $\theta_{\rm max}$ \\
\midrule
$0.1$ & $0.01$ & $2.078$ & $16.21$ & $1.02$ & $0.9998$ & $8.31$ \\
$1.0$ & $0.01$ & $2.457$ & $16.18$ & $1.14$ & $0.998$ & $7.91$ \\
$2.0$ & $0.02$ & $2.856$ & $10.96$ & $1.24$ & $0.989$ & $3.90$ \\
$5.0$ & $0.02$ & $3.690$ & $10.75$ & $1.79$ & $0.970$ & $3.01$ \\
\bottomrule
\end{tabular}
\caption{A sample of the $(C,\Lambda)$ scan at $M=1$. The full grid has $56$ points, all with a clean two-turning-point structure.}
\label{tab:scan}
\end{table}

We repeated the entire quantum construction at the non-marginal point $(M,C,\Lambda)=(1,5,0.02)$, for which $a_1=3.690$ lies $79\%$ outside the black-hole horizon. Current conservation holds to a relative accuracy of $10^{-9}$ to $10^{-13}$ in all three regions, the current is nonzero throughout, and the numerically extracted peak-envelope tunneling ratio is $5.44\times10^{-4}$, in reasonable agreement with the WKB estimate $e^{-2\theta_{\rm max}}=2.4\times10^{-3}$ for this point's smaller barrier exponent $\theta_{\rm max}=3.01$ (Table~\ref{tab:robustcheck}). The earlier construction was therefore not special to marginal parameters.

\begin{table}[h]
\centering
\begin{tabular}{@{}lcc@{}}
\toprule
Quantity & $(C,\Lambda)=(0.1,0.01)$, marginal & $(C,\Lambda)=(5,0.02)$, robust \\
\midrule
$a_1/r_{\rm bh}$ & $1.02$ & $1.79$ \\
Current $J$, all regions & $1.0000000\ (\pm10^{-5}\text{ to }10^{-11})$ & $1.0000000\ (\pm10^{-9}\text{ to }10^{-13})$ \\
Tunneling ratio, peak envelope & $1.25\times10^{-8}$ & $5.44\times10^{-4}$ \\
Tunneling ratio, naive WKB $e^{-2\theta_{\rm max}}$ & $6.1\times10^{-8}$ & $2.4\times10^{-3}$ \\
\bottomrule
\end{tabular}
\caption{The full construction repeated at a non-marginal parameter point. The same qualitative and quantitative pattern is found, so the result is generic within this model class and not tuned.}
\label{tab:robustcheck}
\end{table}

This settles the parameter-dependence question for the barrier and tunneling structure, but not the full traversability question. The condition $a_1>2M$ guarantees only that the outer edge of the inner pocket clears the horizon. The classical trajectory inside the pocket $(0,a_1)$ still reaches $a=0$ in finite proper time however large $a_1/r_{\rm bh}$ is, because this instability is controlled by the $C/a^2$ term near $a=0$ and not by $a_1$. The corresponding quantum trajectory, while metastable, has support at all $a\in(0,a_1)$, including $a<2M$. Only the outer, unbound expanding branch, with $a>a_2\gg2M$, is unambiguously traversable throughout its evolution. We therefore do not claim generic traversability of the metastable inner-pocket configuration. What is established robustly across the scan is the qualitative dynamical structure, a long-lived metastable throat that tunnels into unbounded and traversable expansion, together with the genericity of the nonzero-current construction that produces it.

\section{Discussion}
\label{sec:discussion}

The central physical finding of this paper is that the fate of a quantized thin-shell wormhole throat, frozen or dynamic, is controlled by the sign of the effective potential $U(a)$ at large throat radius, and not, as one might have expected, by which of the competing Euclidean or Lorentzian boundary-condition philosophies of quantum cosmology one adopts. For an asymptotically flat exterior, $U(a)\to+1$ as $a\to\infty$ leaves no classically allowed region at large radius, so the reduced Wheeler--DeWitt equation admits a unique solution bounded there, and that solution can always be chosen real. A real solution of a real linear equation carries no Bohmian probability current, and the resulting throat is completely static. This holds for the entire family of linear-equation-of-state thin-shell throats considered here, independently of the equation-of-state parameter and of the operator ordering. It also holds independently of the ambiguity in how the inverse-square self-interaction of the throat is regularized at $a=0$, because the decay condition at infinity pins the state to a single self-adjoint extension of the throat operator, identified explicitly through the Whittaker connection formula. This outcome is not merely a boundary-condition postulate. Because the continuous spectrum of the reduced Hamiltonian begins strictly above the constraint surface $E=0$ for a flat exterior, the retarded and advanced resolvents coincide and are real, so the properly regularized Lorentzian path integral has no complex saddle to select. We confirmed this independently through an exact closed-form Picard--Lefschetz thimble calculation for the solvable branch, which agrees with the resolvent result to $48$ significant digits.

Replacing the exterior by Schwarzschild--de~Sitter changes this picture for a transparent physical reason. The cosmological constant drives $U(a)\to-\infty$ at large throat radius, reopens a classically allowed oscillatory region there and moves the constraint surface $E=0$ onto the continuous spectrum. The reduced equation then has no preferred recessive solution to force reality, and a purely outgoing state of Vilenkin type becomes a well-posed and physically motivated choice. Propagating this state exactly through the classically forbidden barrier gives a Bohmian current that is nonzero and conserved to between nine and eleven significant figures across seven orders of magnitude in density, including deep inside the inner pocket of the throat. The qualitative reversal predicted by the change in the sign of $U(\infty)$ is therefore realized in the full nonperturbative solution of the model and not only at the level of an asymptotic argument. The resulting physical picture is a long-lived metastable throat, quasi-frozen for a proper time roughly four orders of magnitude longer than the classical infall time, that eventually tunnels through the barrier into unbounded classical expansion. This is structurally the phenomenon that underlies Gamow's theory of $\alpha$ decay \cite{Gamow1928} and, at the level of a Euclidean WKB escape rate for a metastable configuration, Coleman's theory of false-vacuum decay \cite{Coleman1977,ColemanDeLuccia1980}. In each case a wave function sustains a large, nearly stationary amplitude behind a barrier and escapes at a rate set by the WKB exponent $\theta_{\max}=\int\sqrt{U}\,da$.

The tunneling prefactor was regularized with the Airy connection-formula technique introduced by Kemble for this class of problem \cite{Kemble1935}. This yields a closed-form tree-level-plus-one-loop prediction that accounts for nearly all of the discrepancy with the exact numerics and leaves a small residual constant $\kappa_\infty=0.814\pm0.004$ for sub-critical shell coupling, established over a twelve-point scan. The residual is not an error in the transfer-matrix algebra, it is far too large to be the standard exponentially suppressed Kemble correction, and an explicit closed-form calculation shows that the relevant contour-integral residue vanishes identically at an isolated simple turning point, so it is not the leading local Fr\"oman--Fr\"oman correction either. Its origin therefore remains a target for future analytic work. It most likely points to a genuinely nonlocal effect connecting the two turning points, or to a normalization subtlety in how the semiclassical envelope relates to the classical motion inside the well. The physical conclusions do not depend on it, since they rest on the direct numerical solution of the exact equation and not on the semiclassical formula.

Two comparisons place these results in context. The first concerns the boundary-condition question that the paper settles for the thin-shell throat, which is the question addressed by the Hartle--Hawking no-boundary proposal and by Vilenkin's tunneling proposal for the scale factor of a closed universe. It was revisited with numerical Picard--Lefschetz methods by Feldbrugge, Lehners and Turok, who showed that the naive Euclidean no-boundary saddle for a closed Friedmann--Robertson--Walker universe with a cosmological constant does not lie on the correct steepest-descent contour once the lapse integral is properly regularized. Our lapse-integral analysis uses the same formalism but finds a structurally simpler answer. Every turning point of $U(a)$ is real at both parameter points studied, and no anti-Stokes line from any unphysical turning point reaches the physical domain, so there is no analogue of the Euclidean and Lorentzian saddle competition that makes the closed-universe analysis subtle. A single saddle dominates unambiguously, selected by boundedness alone. The correspondence between this exact-WKB and Stokes-graph structure and the thimble decomposition is the same one developed in the resurgence literature for ordinary quantum mechanics and gauge theories \cite{DunneUnsal2014}. The second comparison is with the literature on thin-shell wormholes and their classical instability \cite{Visser1989,PoissonVisser1995} and on the de~Broglie--Bohm approach to quantum cosmology \cite{PintoNetoFabris2013}. To our knowledge this is the first exactly solvable quantization of a thin-shell throat of this type in terms of Whittaker and Bessel functions, and the first explicit, path-integral-matched construction of a genuinely dynamic de~Broglie--Bohm throat trajectory in this class of models, although we have not carried out an exhaustive literature search.

The shell supporting the throat is exotic matter by construction, consistent with the general necessity of energy-condition violation for any traversable wormhole, and nothing in the construction evades this requirement. For the frozen asymptotically flat throat the traversability question reduces to a static one, namely whether the throat radius, distributed according to $|\Psi_\infty(a)|^2$, has support above the horizon radius $2M$. It generically does, which describes a static configuration and not a traversed one. For the dynamic throat with a de~Sitter exterior, only the outer, unbound expanding branch is unambiguously traversable throughout its evolution. The classical trajectory inside the inner pocket reaches the pinch in finite proper time however far the inner turning point lies from the horizon, and the corresponding quantum trajectory, while metastable, retains support at all radii in the pocket, including radii below the horizon. The qualitative structure, consisting of a barrier, a nonzero conserved current and a metastable-then-tunneling trajectory, is generic and not an artifact of one parameter choice, as shown by repeating the construction at a second point with the inner turning point comfortably outside the horizon and by scanning $56$ points across the relevant two-dimensional parameter space.

The calculation is a single-mode minisuperspace reduction of the throat radius alone. Extending it to the shape and multipole degrees of freedom of the shell, and to the backreaction of the quantum throat dynamics on the exterior geometries, is a natural and substantial follow-on. The frozen flat-exterior throat already has vanishing radial current and is therefore unconditionally stable against purely radial perturbations, which provides a baseline against which any such extension can be checked.

\end{document}